\pdfoutput=1
\newif\ifFull
\Fulltrue
\documentclass[runningheads]{llncs}
\usepackage{cite}
\usepackage{amsmath}
\usepackage{graphicx}
\usepackage{color}
\usepackage[pdfencoding=auto]{hyperref}

\usepackage{tikz}
\usetikzlibrary{arrows,arrows.meta,hobby,backgrounds,calc,shapes}
\usepackage{caption}
\usepackage{subcaption}
\usepackage{cleveref}
\usepackage{algpseudocodex}
\usepackage{changepage}
\usepackage[misc]{ifsym}

\hypersetup{
    colorlinks=true,
    linkcolor=blue,
    urlcolor=blue,
    linktoc=all,
    citecolor=blue
}

\crefdefaultlabelformat{{\color{blue}\rmfamily#2#1#3}}
\newcommand{\myref}[1]{{\color{blue}\rmfamily\ref{#1}}}

\newcommand{\IF}{\textbf{if}~}
\newcommand{\THEN}{~\textbf{then}~}
\newcommand{\ELSE}{~\textbf{else}~}
\newcommand{\ELIF}{\textbf{else if}~}
\newcommand{\convexpath}[2]{
[   
    create hullnodes/.code={
        \global\edef\namelist{#1}
        \foreach [count=\counter] \nodename in \namelist {
            \global\edef\numberofnodes{\counter}
            \node at (\nodename) [draw=none,name=hullnode\counter] {};
        }
        \node at (hullnode\numberofnodes) [name=hullnode0,draw=none] {};
        \pgfmathtruncatemacro\lastnumber{\numberofnodes+1}
        \node at (hullnode1) [name=hullnode\lastnumber,draw=none] {};
    },
    create hullnodes
]
($(hullnode1)!#2!-90:(hullnode0)$)
\foreach [
    evaluate=\currentnode as \previousnode using \currentnode-1,
    evaluate=\currentnode as \nextnode using \currentnode+1
    ] \currentnode in {1,...,\numberofnodes} {
  let
    \p1 = ($(hullnode\currentnode)!#2!-90:(hullnode\previousnode)$),
    \p2 = ($(hullnode\currentnode)!#2!90:(hullnode\nextnode)$),
    \p3 = ($(\p1) - (hullnode\currentnode)$),
    \n1 = {atan2(\y3,\x3)},
    \p4 = ($(\p2) - (hullnode\currentnode)$),
    \n2 = {atan2(\y4,\x4)},
    \n{delta} = {-Mod(\n1-\n2,360)}
  in 
    {-- (\p1) arc[start angle=\n1, delta angle=\n{delta}, radius=#2] -- (\p2)}
}
-- cycle
}

\renewcommand{\emph}[1]{\textit{\textbf{#1}}}
\renewcommand{\subsection}[1]{\paragraph{\textbf{#1.}}}
\DeclareMathOperator{\rank}{rank}
\let\epsilon\varepsilon

\begin{document}
\title{Zip-zip Trees: Making Zip Trees More Balanced, Biased, Compact,
or Persistent\thanks{%
  Research at Princeton Univ. was partially supported by 
  a gift from Microsoft.
  Research at Univ. of California, Irvine was supported
  by NSF Grant 2212129.}}
\author{
Ofek Gila\inst{1}\textsuperscript{*}
\orcidID{0009-0005-5931-771X}
% \textsuperscript{(\Letter)}
\and
Michael T. Goodrich\inst{1}\textsuperscript{*}
\orcidID{0000-0002-8943-191X}
\and
Robert E. Tarjan\inst{2}\textsuperscript{*}
\orcidID{0000-0001-7505-5768}
}

\authorrunning{O. Gila et al.}
% \institute{$^1$Univ.~of California, Irvine,
% $^2$Princeton University}
% \date{}   % Uncomment to remove the date if you see it

\institute{University of California, Irvine CA 92697, USA \\
\email{\{ogila, goodrich\}@uci.edu}
\and Princeton University, Princeton NJ 08544, USA \\
\email{ret@cs.princeton.edu}}

% \institute{University of California, Irvine CA 92697, USA \\
% \email{\{ogila, goodrich\}@uci.edu} \and Princeton University, Princeton NJ 08544, USA

% % \email{}
% }


\maketitle

\begin{abstract}
We define simple variants of zip trees, called \emph{zip-zip trees},
which provide several advantages over zip trees, including overcoming a
bias that favors smaller keys over larger ones.
We analyze zip-zip trees theoretically and empirically, showing, e.g.,
that the expected depth of a node in an $n$-node zip-zip tree is 
at most $1.3863\log n-1+o(1)$,
which matches the expected depth of
% whereas the expected
% depth of a node in an (original) zip tree is $1.5\log n+O(1)$.
% Zip-zip trees have the same expected depth as 
treaps and binary search trees built by uniformly random insertions. 
Unlike these other data structures, however,
zip-zip trees achieve their bounds
using only $O(\log\log n)$ bits of metadata per node, w.h.p.,
as compared to the $\Theta(\log n)$ bits per node required by treaps.
In addition, we describe a ``just-in-time'' zip-zip tree variant, which needs
just an expected $O(1)$ number of bits of metadata per node.
Moreover, we can define zip-zip trees to
be strongly history independent, whereas treaps 
are generally only weakly history independent.
We also introduce \emph{biased zip-zip trees}, which have an explicit bias
based on key weights, so the
expected depth of a key, $k$, with weight, $w_k$, is $O(\log (W/w_k))$, where
$W$ is the weight of all keys in the weighted zip-zip tree.
Finally, we show that one can easily make zip-zip trees partially persistent
with only $O(n)$ space overhead w.h.p.
\end{abstract}

\pagestyle{empty}

% \clearpage
\section{Introduction}
A \emph{zip tree} is a type of randomized binary search tree introduced by Tarjan,
Levy, and Timmel~\cite{zip}.
Each node contains a specified key and a small randomly generated \emph{rank}.
Nodes are in symmetric order by key, smaller to larger, 
and in max-heap order by rank.
At a high level, zip trees are similar to 
other random search structures, such as
the \emph{treap} data structure
of Seidel and Aragon~\cite{treaps}, the \emph{skip list} data structure
of Pugh~\cite{skip-lists}, 
and the \emph{randomized binary search tree} (RBST) data structure of 
Mart\'{\i}nez and Roura~\cite{rbst},
but with two advantages:
\begin{enumerate}
\item
Insertions and deletions in zip trees are described in terms of simple ``zip''
and ``unzip'' operations rather than sequences of 
rotations as in treaps and RBSTs, which are arguably more complicated; and 
\item
Like treaps, zip trees organize keys using random ranks, but the ranks used
by zip trees use $\Theta(\log \log n)$ bits each, whereas the 
key labels used by treaps and RBSTs use $\Theta(\log n)$ bits each.
Also, as we review and expand upon, zip trees are topologically isomorphic to
skip lists, but use less space.
\end{enumerate}
% \bob{Delete. This depends on the representation of nodes in a skip list.  Can discuss in more detail in the appropriate section: For example, a zip tree storing $n$ keys has exactly $n$ nodes, whereas a skip
% list holding $n$ keys has an expected $2n$ nodes.}

In addition,
zip trees have a desirable privacy-preservation property with respect
to their
\emph{history independence}~\cite{hartline}.
A data structure is \emph{weakly history independent}
if, for any two sequences of operations $X$ and $Y$ that take
the data structure from initialization to state $A$, 
the distribution over memory after $X$ is
performed is identical to the distribution after $Y$.
Thus, if an adversary observes the final state of the data structure,
the adversary cannot determine the sequence of operations that led to 
that state.
A data structure 
is \emph{strongly history independent}, on the other hand,
if, for any two (possibly empty) sequences of operations
$X$ and $Y$ that take a data structure in state $A$ to state $B$, the distribution over representations
of $B$ after $X$ is performed on a representation, $r$, is identical to the distribution after $Y$ is
performed on $r$.
Thus, if an adversary observes the states of the data structure
at different times,
the adversary cannot determine the sequence of operations that led to 
the second state beyond just what can be inferred from the states themselves.
For example, it is easy to show that skip lists and zip trees are 
strongly history independent, and that treaps and RBSTs are
weakly history independent.\footnote{If the random priorities used in a treap
  are distinct and unchanging for all keys and all time (which occurs
  only probabilistically), then the treap is strongly history independent.}

Indeed, zip trees and skip lists are strongly history independent for
exactly the same reason, 
since Tarjan, Levy, and Timmel~\cite{zip} define zip trees
using a tie-breaking rule for ranks
that makes zip trees isomorphic to skip lists, so that, for instance,
a search in a zip tree would encounter the same keys as would be encountered
in a search in an isomorphic skip list.
% \red{Bob: I would point out here that we (the authors of the original paper) made this choice deliberately, to make zip trees isomorphic to skip lists.}
This isomorphism between zip trees and skip lists has
a potentially undesirable property,
however,
in that there is an inherent bias in a zip tree that favors smaller
keys over larger keys.
For example, 
as we discuss,
the analysis from Tarjan, Levy, and Timmel~\cite{zip}
implies that
the expected depth of the smallest key in an (original) 
zip tree is $0.5\log n$ whereas the
expected depth of the largest key is $\log n$.  
Moreover, this same analysis implies
that the expected depth for 
any node in a zip tree is at most $1.5\log n+O(1)$, whereas 
Seidel and Aragon~\cite{treaps} show that the expected depth of any node
in a treap is at most $1.3863\log n+1$, and 
Mart\'{\i}nez and Roura~\cite{rbst} prove a similar result
for RBSTs.

As mentioned above,
the inventors of zip trees chose their tie-breaking rule to 
provide an isomorphism between zip trees and skip lists.  
But one may ask if there is a (hopefully simple) 
modification to the tie-breaking rule for zip trees that makes 
them more balanced
for all keys, ideally while still maintaining the property that 
they are strongly history independent and that the metadata
for keys in a zip tree requires only $O(\log\log n)$ bits per key w.h.p.

Note that the structure of zip trees is identical in structure to the \emph{skip list tree}, independently discovered by Erickson two years prior \cite{erickson_treaps_2017}. Skip list trees perform insertions and deletions using rotations rather than through the zip and unzip operations of zip trees.

In this paper, we show how to improve the balance of nodes
in zip trees by a remarkably simple change to its tie-breaking rule
for ranks. 
Specifically, we describe and analyze a zip-tree variant we call
\emph{zip-zip trees},
% \begin{itemize}
% \item
% \emph{zig-zag zip trees}:
% In this zip tree variant, we
% break rank ties by biasing to the left if the rank is odd and 
% to the right if the rank
% is even. Like the original zip trees, 
% zig-zag zip trees are strongly history
% independent.
% \todo{I think these will have the same expected depth as original 
% zip trees, but without the bias for small keys. Even so,
% I don't think we need this zip tree in our paper anymore.}
% \item
% \emph{uniform zip trees}:
% In this zip tree variant, we have defined a zip tree where the ranks
% are integers independently chosen 
% uniformly at random in the range $[1,N]$, but 
% with all other aspects of their construction the same as for original zip 
% trees.
% We say that a uniform zip tree is \emph{pure} if all its ranks are distinct;
% note that a pure uniform zip tree is isomorphic to a random treap, hence,
% its expected height is $2\ln n$ in this case.
% \todo{I think the expected depth of a uniform zip tree 
% is $2\ln n$ regardless of 
% whether it is pure, which occurs with high probability if $N\ge n^3$.}
% Unlike a random treap, a uniform zip tree is constructed using zip and unzip
% operations rather than rotations, which are more complicated.
% Regardless of whether it is pure or not, a uniform zip tree is strongly
% history independent.
% \item
% \emph{zip-zip trees:}
in which
we give each key a rank pair, $r=(r_1,r_2)$, such that $r_1$ is chosen from
a geometric distribution as in the original definition of zip trees, and $r_2$
is an integer chosen uniformly at random, e.g., in the range $[1,\log^c n]$, 
for $c\ge 3$. We build a zip-zip tree 
just like an original zip tree, but with these rank pairs as its ranks, 
ordered and compared lexicographically.  
We also consider a just-in-time (JIT) variant of zip-zip trees, where we
build the secondary $r_2$ ranks bit by bit as needed to break ties.
Just like an original zip tree, 
zip-zip trees (with static secondary ranks) 
are strongly history independent,
and, in any variant,
each rank in a zip-zip tree requires only $O(\log\log n)$ bits w.h.p.
Nevertheless, as we show (and verify experimentally),
the expected depth of any node in a zip-zip tree storing $n$ keys is 
at most $1.3863\log n-1+o(1)$, whereas the expected depth of a node
in an original zip tree is $1.5\log n+O(1)$, as mentioned above.
We also show (and verify experimentally) 
that the expected depths of the smallest and largest keys
in a zip-zip tree are the same---namely, they both are at most 
$0.6932\log n + \gamma + o(1)$, where $\gamma=0.577721566\ldots$ is the 
Euler-Mascheroni constant.

% \item
% \emph{dynamically-random zip trees:}
% In this zip tree variant,
% we break rank ties dynamically at random so that each 
% subtree of nodes with the same
% rank is a random binary tree where all such trees are equally likely. 
% We do this by making random choices during zip and unzip operations
% that depend on us maintaining for each node the number of its descendents 
% that have the same rank, in a fashion reminiscent 
% of the random binary search tree 
% (RBST) of Mart\'{\i}nez and Roura~\cite{rbst}.
% This results in a zip tree that is weakly history independent, but, as
% we show, it will have the same expected depth (for all keys) as
% for random treaps. 
% \todo{I don't think we need this zip tree in our paper anymore.}
% \end{itemize}
% \bob{Another way to save space that we didn't explore is to store rank differences instead of ranks.  Most rank differences will be small.} \red{Ofek: Rank differences would be small for the geometric distribution, but not for the uniform distribution, right? So the total space would still be $\log \log n$ I think. But worth mentioning for geometric distribution.}  \bob{Let's mention this somewhere in the text if we can find a place for it.}

In addition to showing how to make zip trees more balanced, by
using the zip-zip tree tie-breaking rule,
we also describe how to make them more biased for weighted keys.
Specifically, we study how to store weighted
keys in a zip-zip tree, giving us the following 
variant (which can also be implemented for the original zip-tree
tie-breaking rule):
\begin{itemize}
\item
\emph{biased zip-zip trees}:
These are a biased version of zip-zip trees, 
which support searches with expected performance bounds that 
are logarithmic in $W/w_k$, where $W$ is the total weight of all keys 
in the tree and $w_k$ is the weight of the search key, $k$.
\end{itemize}
Biased zip-zip trees
can be used in simplified versions of the link-cut tree
data structure of Sleator and Tarjan~\cite{link-cut} for dynamically
maintaining arbitrary trees, which has many applications, 
e.g., see Acar~\cite{acar}.

Zip-zip trees and biased zip-zip trees have
% The zip-tree variants we discuss in this paper utilize
only $O(\log\log n)$ bits of metadata per key w.h.p. 
(assuming polynomial weights
in the weighted case) and are strongly history independent 
.
The just-in-time (JIT) variant utilizes only $O(1)$ bits of metadata per operation w.h.p.
but lacks history independence.
% (unless the secondary ranks are dynamically extended as needed to break ties). 
Moreover, if zip-zip trees are implemented using the tiny pointers 
technique of Bender, Conway, Farach-Colton, Kuszmaul, 
and Tagliavini~\cite{tiny-pointers}, then all of the non-key data
used to implement such a tree requires just 
$O(n\log\log n)$ bits overall w.h.p.

\subsection{Additional Prior Work}
Before we provide our results, let us briefly review some additional
related prior work.
Although this analysis doesn't apply to treaps or RBSTs,
Devroye~\cite{devroye1986note,devroye1987branching} showed 
that the expected height of a
randomly-constructed binary search tree tends to $4.311\log n$ in
the limit, which tightened a similar earlier result
of Flajolet and Odlyzko~\cite{flajolet1982average}.
Reed~\cite{reed2003height} tightened this bound even further, showing
that the variance of the height of a randomly-constructed binary search tree is $O(1)$.
Eberl, Haslbeck, and Nipkow~\cite{eberl2018verified} showed that this
analysis also applies to treaps and RBSTs, with respect to their expected
height.
Papadakis, Munro, and Poblete~\cite{papadakis1992average} provided
an analysis for the expected search cost in a skip list, showing
the expected cost is roughly $2\log n$.

With respect to weighted keys, 
Bent, Sleator, and Tarjan~\cite{bent1985biased} introduced 
a \emph{biased search tree} data structure,
for storing a set, $\mathcal{K}$, of $n$ weighted keys, with
a search time of $O(\log (W/w_k))$, where $w_k$ is the weight
of the search key,~$k$,
and $W=\sum_{k\in \mathcal{K}} w_k$.
Their data structure is not history independent, however.
Seidel and Aragon~\cite{treaps} provided a weighted version of treaps,
which are weakly history independent and have expected
$O(\log (W/w_k))$ access times, but weighted treaps 
have weight-dependent key labels that use exponentially more bits than are 
needed for weighted zip-zip trees.
% new citation below
Afek, Kaplan, Korenfeld, Morrison, and Tarjan~\cite{afek_cb_2014} provided a fast
concurrent  self-adjusting biased search tree when the weights are access 
frequencies. Zip trees and by extension zip-zip trees would similarly work well in a
concurrent setting, since most updates affect only the bottom of the tree, and updates can be done purely top down, although such
an implementation is not explored in this paper.
Bagchi, Buchsbaum, and Goodrich~\cite{bagchi2005biased} introduced randomized 
\emph{biased skip lists}, which are strongly history independent
and in which the expected time to access a key, $k$,
is likewise $O(\log (W/w_k))$.
Our weighted zip-zip trees
are analogous to biased skip lists, but use less space.
% A biased zip-zip tree storing $n$ keys 
% has exactly $n$ nodes, whereas its dual biased skip list has an expected 
% number of nodes equal to $2n+2\sum_{k\in\mathcal{K}} \log w_k$.
% This space difference can be significant, 
% depending on the distribution of weights.
% For instance, if a constant fraction of a set of keys have a weight
% of at least $n^\epsilon$, for some constant, $\epsilon>0$, then a biased skip
% list for this set of keys will use $\Omega(n\log n)$ space.
% Thus, the weighted case can provide asymptotic 
% benefits for using  biased zip-zip trees over biased skip lists.

\section{A Review of Zip Trees}
In this section, we review the (original) zip tree data structure
of Tarjan, Levy, and Timmel~\cite{zip}.

\subsection{A Brief Review of Skip Lists}
We begin by reviewing a related structure, namely, the
\emph{skip list} structure of Pugh~\cite{skip-lists}.
Let $\log{n}$ denote the base-two logarithm.
A skip list is a hierarchical, linked collection of sorted 
lists that is constructed using randomization.
All keys are stored in level 0, and, for each key, $k$, 
in level $i\ge0$, we include $k$ in the list in level $i+1$
if a random coin flip (i.e., a random bit) 
is ``heads'' (i.e., 1), which occurs with probability $1/2$ 
and is independent of all other coin flips.
Thus, we expect half of the keys on level $i$ to also appear in level
$i+1$.
In addition, every level includes a node that stores a key,
$-\infty$, that is less than every other key, and
a node that stores a key,
$+\infty$, that is greater than every other key.
The highest level of a skip list is the smallest $i$ such that the list at
level $i$ only stores $-\infty$ and $+\infty$.
(See \Cref{fig:skip}.)
The following theorem follows from well-known properties of skip lists.

\begin{theorem}
\label{thm:skip-max}
Let $S$ be a skip list built from $n$ distinct keys. 
% \begin{enumerate}
% \item 
The probability that 
the height of $S$ is more than $\log n + f(n)$ is
at most $2^{-f(n)}$, 
for any monotonically increasing function $f(n)>0$.
% \item
% The expected number of nodes in $S$ is $2n$.
% \end{enumerate}
\end{theorem}

\setlength{\textfloatsep}{10pt plus 1.0pt minus 2.0pt}

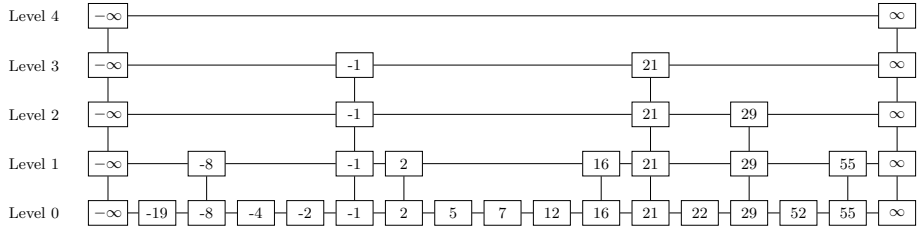
\begin{figure}[t!]
    \centering
    \setlength{\belowcaptionskip}{0pt}
    \resizebox{\linewidth}{!}{\begin{tikzpicture}[n/.style = {rectangle, draw, minimum width = 0.75cm, minimum height=0.5cm}]
    %                  -1                   21,
    %                  -1,                  21,     29,
    %      -8,         -1, 2,           16, 21,     29,     55
    % -19, -8, -4, -2, -1, 2, 5, 7, 12, 16, 21, 22, 29, 52, 55

    \def\maxheight{4}

    \def\skiplist{
        {-100/\maxheight},
        {-19/0},
        {-8/1},
        {-4/0},
        {-2/0},
        {-1/3},
        {2/1},
        {5/0},
        {7/0},
        {12/0},
        {16/1},
        {21/3},
        {22/0},
        {29/2},
        {52/0},
        {55/1},
        {100/\maxheight}%
    }

    % Draw level labels
    \foreach \y in {0,...,\maxheight}{
        \node (\y) at (-0.5, \y) {Level \y};
    }
    
    % Draw all the nodes
    \foreach \key\height [count=\i] in \skiplist{
        \foreach \y in {0,...,\height}{
            \node[n] (\i\y) at (\i, \y) {
                \ifthenelse{\key = -100}{$-\infty$}{\ifthenelse{\key = 100}{$\infty$}{\key}}
            };
        }
    }

    % Connect each level
    \foreach \y in {0,...,\maxheight}{
        \pgfmathtruncatemacro\lasti{1}
        \foreach \key\height [count=\i] in \skiplist{
            \ifthenelse{\i = 1 \OR \height < \y}{}{
                \draw[-] (\lasti\y) -- (\i\y);
                \global\let\lasti=\i;
            }

            \ifthenelse{\y = 0 \OR \height < \y}{}{
                \pgfmathtruncatemacro\lowery{\y-1};
                \draw[-] (\i\lowery) -- (\i\y);
            }
        }
    }
\end{tikzpicture}}

    \caption{\label{fig:skip} An example skip list.}
\end{figure}

\setlength{\textfloatsep}{10pt}

\begin{proof}
Note that the highest level in $S$ is determined by 
the random variable $X=\max\{X_1,X_2,\ldots,X_n\}$, where each $X_i$
is an independent geometric random variable with success probability $1/2$.
Thus, for any $i=1,2,\ldots, n$, 
\[
\Pr(X_i> \log n + f(n)) < 2^{-(\log n + f(n))} = 2^{-f(n)}/n;
\]
By a union bound, $\Pr(X>\log n+f(n))<2^{-f(n)}$.
% 
% For claim~2, note that, by the way a skip-list is constructed,
% the number of nodes, $N$, in $S$ can be characterized as $\sum_{i=1}^n X_i$,
% where each $X_i$ is a geometric random variable with success probability $1/2$.
% Thus, $E[N]=\sum_{i=1}^n E[X_i] = 2n$.
\qed\end{proof}

% \vspace*{-10pt}

\subsection{Zip Trees and Their Isomorphism to Skip Lists}
We next review the definition of the (original) zip 
tree data structure~\cite{zip}.
A zip tree is a binary search tree in which nodes are 
max-heap ordered according to random \emph{ranks}, with 
ties broken in favor of smaller keys, so that the parent of a node has
rank greater than that of its left child and 
no less than that of its right child~\cite{zip}. 
The rank of a node is drawn from a geometric distribution with 
success probability $1/2$, starting from a rank $0$, so 
that a node has rank $k$ with probability $1/2^{k+1}$.

As noted by Tarjan, Levy, and Timmel~\cite{zip},
there is a natural isomorphism between a 
skip-list, $L$,  and a zip tree, $T$, where 
$L$ contains a key $k$ in its level-$i$ list if and only if
$k$ has rank at least $i$ in $T$.
That is, the rank of a key, $k$, in $T$ equals the highest level in $L$ that
contains $k$.
See \Cref{fig:zip-tree}.
\ifFull
Incidentally, this isomorphism is 
topologically identical to a duality between skip lists and binary search
trees observed earlier by
Dean and Jones~\cite{duality},
but the constructions of Dean and Jones are for binary search trees 
that involve rotations to maintain balance
and have different metadata than zip trees, so,
apart from the topological similarities, the analyses of Dean
and Jones don't
apply to zip trees.
\fi
% As we review in an appendix, insertion and deletion in a zip tree
% are done by simple ``unzip'' and ``zip'' operations, and these same algorithms
% also apply to the variants we discuss in this paper, with the only 
% difference being the way we define ranks.

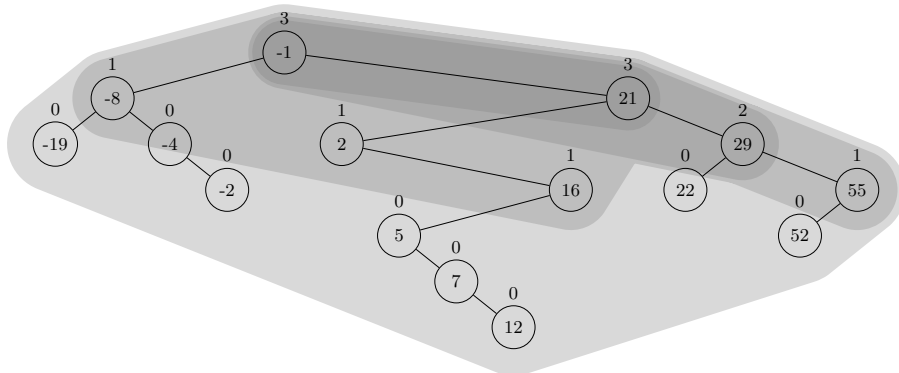
\begin{figure}[b!]
    \centering
    \resizebox{\linewidth}{!}{\begin{tikzpicture}[n/.style = {circle, draw, minimum width = 0.75cm, minimum height=0.5cm}]
    \def\ziptree{
        {-19/0/2},
        {-8/1/1},
        {-4/0/2},
        {-2/0/3},
        {-1/3/0},
        {2/1/2},
        {5/0/4},
        {7/0/5},
        {12/0/6},
        {16/1/3},
        {21/3/1},
        {22/0/3},
        {29/2/2},
        {52/0/4},
        {55/1/3}%
    }

    % Draw all the nodes
    \foreach \key\rank\height [count=\i] in \ziptree{
        \pgfmathsetmacro\x{\i};
        \pgfmathsetmacro\y{-0.8*\height};
        % \pgfmathsetmacro\x{\i};
        % \pgfmathsetmacro\y{-\height};
        \node[n,label=\rank] (\key) at (\x, \y) {\key};
    }

    \draw[-] (-1) -- (-8) -- (-19)
             (-8) -- (-4) -- (-2)
             (-1) -- (21)
             (12) -- (7) -- (5) -- (16) -- (2) -- (21) -- (29) -- (55) -- (52)
             (29) -- (22);

    \begin{pgfonlayer}{background}
        \fill[gray,opacity=0.3] \convexpath{-1,21}{16pt};
        \fill[gray,opacity=0.3] \convexpath{-1,21,29}{18pt};
        \fill[gray,opacity=0.3] \convexpath{-1,21,29,55,29,-1,21,16,2,-8}{20pt};
        \fill[gray,opacity=0.3] \convexpath{-19,-8,-1,21,55,52,12}{24pt};
    \end{pgfonlayer}
\end{tikzpicture}}

    \caption{\label{fig:zip-tree} An example zip tree, corresponding
to the skip list in \Cref{fig:skip}. }
\end{figure}

% \begin{figure}[hbt]
% \centering
% \includegraphics[scale=.5]{figs/zip-tree.png}
% \caption{\label{fig:zip-tree} An example zip tree, corresponding
% to the skip list in \Cref{fig:skip}.
% \todo{This is stolen from Dean and Jones. Make a new figure, which has
% ranks, not rank differences, and doesn't include $r$.}
% }
% \end{figure}

An advantage of a zip tree, $T$, over its isomorphic skip list, $L$,
is that $T$'s space usage
is roughly half of that of $L$, and $T$'s
search times are also better. 
Nevertheless, there is a potential undesirable property of zip trees,
in that an original zip tree is biased towards smaller keys,
as we show in the following.

\begin{theorem}
\label{thm:zip-smallest-largest}
Let $T$ be an (original) zip tree storing $n$ distinct keys.
Then the expected depth of the smallest key is 
$0.5\log n + O(1)$, whereas the expected depth of the largest
key is $\log n + O(1)$.
\end{theorem}
\begin{proof}
The bound for the largest (respectively smallest) key follows immediately
from Lemma 3.3 (respectively Lemma 3.4) 
from Tarjan, Levy, and Timmel~\cite{zip} and the fact
that the expected largest rank in $T$
is at most $\log n+O(1)$.
\qed\end{proof}

That is, the expected depth of the largest key in an original zip tree
is twice that of the smallest key.
% In addition, this behavior is also seen in our experiments.
This bias also carries over into
the bound of Tarjan, Levy, and Timmel~\cite{zip} on
the expected depth of a node in an original zip tree,
which they show is at most
$1.5\log n + O(1)$.
In contrast, the expected depth of a node in a 
treap or randomized binary search tree is at most
$1.39\log n + O(1)$~\cite{treaps,rbst}.
% We show in the next section how to achieve this expected
% depth bound for a type of zip tree.

\subsection{Insertion and Deletion in Zip Trees and Zip-zip Trees}

Insertion and deletion in a zip tree
are done by simple ``unzip'' and ``zip'' operations. These algorithms
also work for the variants we discuss in this paper, with the only 
difference being in the way we define ranks.

To insert a new node $x$ into a zip tree, we search for $x$ in the tree 
until reaching the node $y$
that $x$ will replace, namely the node $y$ such that 
$y.rank \le x.rank$, with strict inequality if $y.key < x.key$. 
From $y$, we follow the rest of the search path for $x$, emph{unzipping} it by 
splitting it into a path, $P$,
containing each node with key less than $x.key$ and a 
path, $Q$, containing each node with key greater
than $x.key$ (recall that we assume keys are distinct)~\cite{zip}.  The top node on $P$ (respectively $Q$) becomes the left (respectively right) child of the node to be inserted, which itself replaces $y$ as a child of its parent.  
To delete a node $x$, we perform the inverse operation: 
We do a search to find $x$, and let $P$ and $Q$ be
the right spine of the left subtree of $x$ and the left spine 
of the right subtree of $x$, respectively. Then we zip $P$ and $Q$ together to
form a single path $R$, by merging them from top to bottom in 
non-increasing rank order, breaking
a tie in favor of the smaller key~\cite{zip}.  The top node of $R$ replaces $x$ as a child of its parent. 
% For completeness, we include pseudo-code for the insert
% and delete operations for a zip tree in
See \Cref{fig:zip}.
Pseudo-code is provided in \Cref{sec:pseudo}.

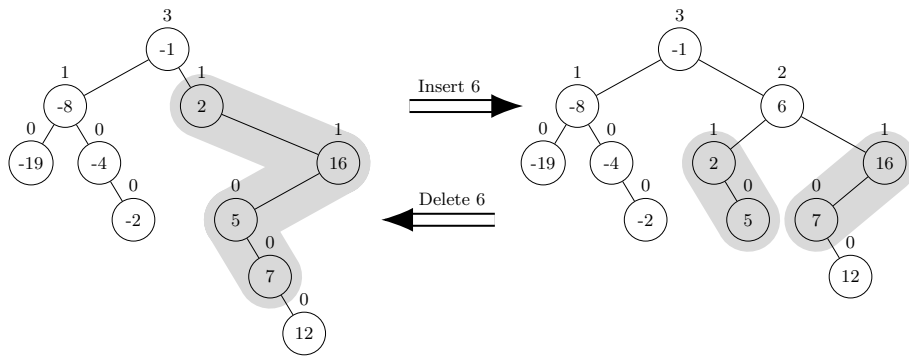
\begin{figure}[hbt]
    \centering
    \resizebox{\linewidth}{!}{\begin{tikzpicture}[n/.style = {circle, draw, minimum width = 0.75cm, minimum height=0.5cm},a/.style = {line width=1pt, double distance=5pt,
             arrows = {-Latex[length=0pt 4 .5]}}]
    \def\ziptreeone{
        {-19/0/2},
        {-8/1/1},
        {-4/0/2},
        {-2/0/3},
        {-1/3/0},
        {2/1/1},
        {5/0/3},
        {7/0/4},
        {12/0/5},
        {16/1/2}%
    }

    % Draw all the nodes
    \foreach \key\rank\height [count=\i] in \ziptreeone{
        \pgfmathsetmacro\x{0.6*\i};
        \node[n,label=\rank] (\key) at (\x, -\height) {\key};
    }

    \draw[-] (-1) -- (-8) -- (-19)
             (-8) -- (-4) -- (-2)
             (12) -- (7) -- (5) -- (16) -- (2) -- (-1);

    \begin{pgfonlayer}{background}
        \fill[gray,opacity=0.3] \convexpath{2,16,5,7,5,16}{16pt};
    \end{pgfonlayer}

    \draw [a] (7.25,-1) -- node[label={[xshift=-0.3cm]Insert 6}] {} (9.25,-1);
    \draw [a] (8.75,-3) -- node[label={[xshift=0.25cm]Delete 6}] {} (6.75,-3);

    \def\ziptreetwo{
        {-19/0/2},
        {-8/1/1},
        {-4/0/2},
        {-2/0/3},
        {-1/3/0},
        {2/1/2},
        {5/0/3},
        {6/2/1},
        {7/0/3},
        {12/0/4},
        {16/1/2}%
    }

    % Draw all the nodes
    \foreach \key\rank\height [count=\i] in \ziptreetwo{
        \pgfmathsetmacro\x{0.6*\i+9};
        \node[n,label=\rank] (\key) at (\x, -\height) {\key};
    }

    \draw[-] (6) -- (-1) -- (-8) -- (-19)
             (-8) -- (-4) -- (-2)
             (5) -- (2) -- (6) -- (16) -- (7) -- (12);

    \begin{pgfonlayer}{background}
        \fill[gray,opacity=0.3] \convexpath{2,5}{16pt};
        \fill[gray,opacity=0.3] \convexpath{16,7}{16pt};
    \end{pgfonlayer}
\end{tikzpicture}}
    
    \caption{\label{fig:zip} How insertion in a zip tree is done via unzipping and  deletion is done via zipping. }
\end{figure}

\section{Zip-zip Trees}
In this section, we define and analyze the zip-zip tree data structure.

\subsection{Uniform Zip Trees}
As a warm-up, let us first 
define a variant of the original zip tree, called
the \emph{uniform zip tree}.  This is a zip tree in which the rank of 
each key is a random integer drawn independently from a uniform
distribution over a suitable range.
We perform insertions and deletions in a uniform zip tree exactly as in 
an original zip tree, except that rank comparisons are done using
these uniform ranks rather than using ranks drawn from 
a geometric distribution.
If there are no rank ties that occur during its construction, a uniform zip tree is a treap~\cite{treaps}.
But if a rank tie occurs,
we resolve it using the tie-breaking rule for a
zip tree, rather than doing a complete tree rebuild,
as is done for a treap~\cite{treaps}.
% Thus, a uniform zip tree is strongly history independent, but it 
% still requires $\Theta(\log n)$ bits of metadata per key, like a treap.
We introduce uniform zip trees only as a stepping stone to our
definition of zip-zip trees, which we give next.

\subsection{Zip-zip Trees}
A \emph{zip-zip tree} is a zip tree in which we define the rank of 
each key to be a pair, $r=(r_1,r_2)$, where $r_1$ is drawn independently 
from a geometric distribution with success probability $1/2$ 
(as in original zip trees) and $r_2$ is an integer drawn independently from
a uniform distribution on the interval $[1,\log^c n]$, for $c\ge3$.
We perform insertions and deletions in a zip-zip tree exactly as in 
an original zip tree, except that rank comparisons are done lexicographically
based on the $(r_1,r_2)$ pairs.
That is, we perform an update operation focused primarily on the $r_1$ ranks,
as in an original zip tree, but we break ties by reverting to 
$r_2$ ranks. And if we still get a rank
tie for two pairs of ranks, then we break these ties as in original zip trees, biasing in favor of smaller keys. 
As we shall show, such ties occur with such low probability that they
don't significantly impact the expected depth of any node
in a zip-zip tree. 
This also implies that 
the expected depth of the smallest key in a zip-zip tree is the
same as for the largest key.

Let $x_i$ be a node in a zip-zip tree, $T$.
Define the \emph{$r_1$-rank group} of $x_i$
as the connected subtree of $T$ containing 
all nodes with the same $r_1$-rank as $x_i$.
That is, each node in $x_i$'s $r_1$-rank group
has a rank tie with $x_i$
when comparing ranks with just the first rank coordinate, $r_1$.

\begin{lemma}
\label{lem:zipzip}
The $r_1$-rank group for any node, $x_i$, in a zip-zip tree
is a uniform zip tree defined using $r_2$-ranks.
\end{lemma}
\begin{proof}
The proof follows immediately from the definitions.
\qed\end{proof}

Incidentally,
\Cref{lem:zipzip} is the motivation for the name ``zip-zip tree,''
since a zip-zip tree can be viewed as a zip tree comprised of little zip trees.
Moreover, this lemma immediately implies that a zip-zip tree is strongly
% history independent, since a zip tree is strongly history independent.
history independent, since both zip trees and uniform zip trees are strongly history independent.

See \Cref{fig:zipzip}.

\begin{figure}[t]
    \centering
    \resizebox{\linewidth}{!}{\begin{tikzpicture}[n/.style = {circle, draw, minimum width = 0.75cm, minimum height=0.5cm}]
    \def\ziptree{
        {-19/0/33/3},
        {-8/1/26/2},
        {-4/0/31/3},
        {-2/0/1/4},
        {-1/3/13/1},
        {2/1/1/3},
        {5/0/23/5},
        {7/0/46/4},
        {12/0/13/5},
        {16/1/49/2},
        {21/3/31/0},
        {22/0/21/2},
        {29/2/20/1},
        {52/0/2/3},
        {55/1/38/2}%
    }

    % Draw all the nodes
    \foreach \key\rankg\ranku\height [count=\i] in \ziptree{
        \pgfmathsetmacro\x{\i};
        \pgfmathsetmacro\y{-0.8*\height};
        % \pgfmathsetmacro\x{\i};
        % \pgfmathsetmacro\y{-\height};
        \node[n,label={(\rankg,\ranku)}] (\key) at (\x, \y) {\key};
    }

    \draw[-] (-19) -- (-8) -- (-4) -- (-2)
             (-8) -- (-1) -- (21) -- (29) -- (55) -- (52)
             (5) -- (7) -- (12) (7) -- (2) -- (16) -- (-1)
             (22) -- (29);

    \begin{pgfonlayer}{background}
        \fill[gray,opacity=0.3] \convexpath{5,7,12}{28pt};
        \fill[gray,opacity=0.3] \convexpath{-4,-2}{28pt};
        \fill[gray,opacity=0.3] \convexpath{2,16}{28pt};
        \fill[gray,opacity=0.3] \convexpath{-1,21}{28pt};
    \end{pgfonlayer}
\end{tikzpicture}}
    \caption{\label{fig:zipzip} A zip-zip tree, with
     each node labeled with its $(r_1,r_2)$ rank. Each shaded
     subtree is an $r_1$-rank group defining a uniform zip tree based
     on $r_2$ ranks.}
\end{figure}
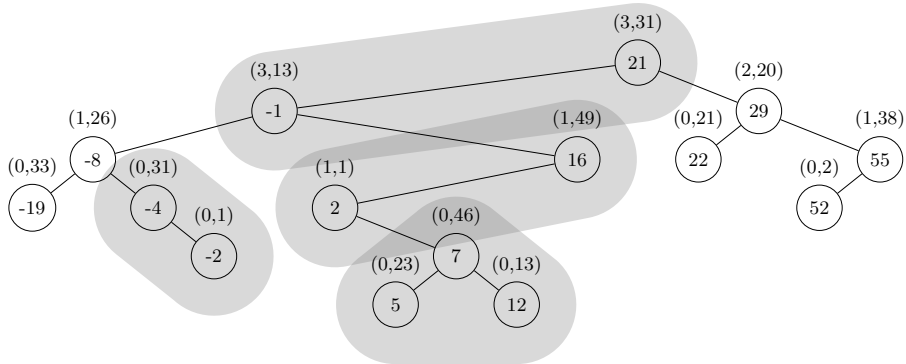

\begin{lemma}
\label{lem:rank-groups}
The number of nodes in an $r_1$-rank group in a zip-zip tree, $T$
storing $n$ keys has expected value $2$ and is at most $2\log n$
with high probability.
\end{lemma}
\begin{proof}
Consider the smallest node of a particular $r_1$-rank group and begin an in-order traversal.
Any node with smaller rank appears beneath the group without affecting it, while any node with higher rank stops the traversal.
The $r_1$-rank group is defined as all nodes encountered along the traversal sharing the same rank as $u$.
The set of nodes
% Any node encountered in the traversal with rank $\ge r_1$ can be considered to have rank as a random variable
% By the isomorphism between zip trees and skip lists, the set of nodes
in an $r_1$-rank group in $T$ is a sequence of consecutive nodes with rank exactly $r_1$
% in a level-$r_1$ list in the isomorphic skip list, none of which is in the level-$(r_1+1)$
% list.
in an in-order traversal starting from a rank-$r_1$ node, stopping when a node is encountered with greater rank.
Thus the number of nodes, $X$, in an $r_1$-rank group is a random variable
drawn from a geometric distribution with success probability $1/2$; hence,
$E[X]=2$ and $X$ is at most $2\log n$ with probability 
at least $1-1/n^2$.
Moreover, by a union bound, all the $r_1$-rank groups in $T$ have size
at most $2\log n$ with probability
at least $1-1/n$.
\qed
\end{proof}

We can also define a variant of a zip-zip tree
that is not history independent
but that uses only $O(1)$ bits of metadata per key in expectation.

\subsection{Just-in-Time Zip-zip Trees}
In a \emph{just-in-time (JIT) zip-zip tree}, 
we define the $(r_1,r_2)$ rank pair
for a key, $x_i$, so that $r_1$ is (as always) drawn independently 
from a geometric distribution with success probability $1/2$, but where
$r_2$ is an initially empty string of random bits.
If at any time during an update in a JIT zip-zip tree, there is a tie
between two rank pairs, $(r_{1,i},r_{2,i})$ and $(r_{1,j},r_{2,j})$,
for two keys, $x_i$ and $x_j$, respectively, then we independently add
unbiased random bits, one bit at a time, to $r_{2,i}$ and $r_{2,j}$
until $x_i$ and $x_j$ no longer have a tie in their rank pairs, where
$r_2$-rank comparisons are done by viewing the binary strings as binary
fractions after a decimal point.

% \todo{Maybe mention that they must have the same number of bits before comparing (except optimization where $r_{2,i}$ has fewer bits but is still greater than $r_{2,j}$)}

Note that the definition of an $r_1$-rank group is the same for
JIT zip-zip trees and (standard) zip-zip trees.
Rather than store $r_1$-ranks explicitly, however, we store them as 
a difference between the $r_1$-rank of a node and the $r_1$-rank
of its parent (except for the root).
% Thus, by Lemma~\ref{lem:rank-groups}, the expected number of bits 
% in the $r_2$ field of a node in a JIT zip-zip tree is at most $2$ and
% it is $O(\log\log n)$ w.h.p.
% 
Moreover, by construction, each $r_1$-rank group in a JIT zip-zip tree
is a treap; hence, a JIT zip-zip tree is topologically isomorphic to
a treap.
% even though each rank in a JIT zip-zip tree is represented
% using only $O(\log\log n)$ bits w.h.p., rather than the $\Theta(\log n)$
% bits required for priorities in a treap.
% We prove the following theorem in an appendix.

\begin{theorem}
\label{thm:jit}
Let $T$ be a JIT zip-zip tree resulting from
$n$ update operations starting from an initially empty tree.
The expected number of bits of rank metadata in any non-root node in $T$
is $O(1)$, and the number of bits required for 
all the rank metadata in $T$ is $O(n)$~w.h.p.
\end{theorem}

To prove this, we use the following lemma:

\begin{lemma}
\label{lem:sum}
Let $X$ be the sum of $n$ independent geometric random variables with success
probability $1/2$. Then, for $t\ge 2$,
\[
\Pr(X>(2+t)n)\le e^{-tn/10} .
\]
\end{lemma}
\begin{proof}
The proof follows immediately
by a Chernoff bound for a sum of $n$ independent geometric random variables
(see, e.g., Goodrich and Tamassia~\cite[pp.~555--556]{goodrich2015algorithm}).
\qed
\end{proof}

Using this lemma, we can prove that JIT zip-zip trees use $O(n)$ total metadata with high probability.

\begin{proof}[of \Cref{thm:jit}]
The set of nodes
in an $r_1$-rank group in $T$ is a sequence of consecutive nodes with rank exactly $r_1$
in an in-order traversal starting from a rank-$r_1$ node, stopping when a node is encountered with greater rank.
All the nodes in this group require $O(1)$ bits to store their $r_1$ rank difference except the root, $v$.
Assuming that the root of this rank group is not the root of the tree, this group has a parent $u$ with rank $r_1' > r_1$. 
% By the isomorphism between zip trees and skip lists, the set of nodes
% in an $r_1$-rank group in $T$ is a sequence, $L$, of consecutive nodes
% in a level-$r_1$ list in the skip list that are not in the level-$(r_1+1)$
% list.
% Thus, since $v$ is not the root, there is a node, $u$, 
% that is the immediate
% predecessor of the first node in $L$ in the level-$r_1$ list in the skip list,
% and there is a node, $w$,
% that is the immediate
% successor of the last node in $L$ in the level-$r_1$ list in the skip list.
% Moreover, both $u$ and $v$ are in the level-$(r_1+1)$ list in the skip list,
% and (since $v$ is not the root) it cannot be the case that 
% $u$ stores the key $-\infty$ and $w$ stores the key $+\infty$.
% As an over-estimate and to avoid dealing with dependencies, 
% we will consider the $r_1$-rank differences determined by
% predecessor nodes (like $u$) separate from 
% the $r_1$-rank differences determined by
% successor nodes (like $w$).
% Let us focus on predecessors, $u$, and
% suppose $u$ does not store $-\infty$. Let $r_1'>r_1$ 
% be the highest level in the skip list where $u$ appears.
% Then the
The
difference between the $r_1$-rank of $v$ and its parent is
% at most $r_1'-r_1$.
$r_1'-r_1$.
% That is, this rank difference is at most a random variable 
That is, this rank difference is a random variable
that is drawn from a geometric distribution
with success probability $1/2$ (starting at level $r_1+1$); hence,
its expected value is at most $2$.
Further, for similar reasons, the sum of all the $r_1$-rank differences
% for all nodes in $T$ that are determined because of a predecessor node 
for all nodes in $T$ that are roots of their rank groups while not being the global root
(like $u$) can be bounded by the sum, $X$, of $n$
independent geometric random variables with success probability $1/2$.  
(Indeed, this is also an over-estimate, since a $r_1$-rank difference
for a parent in the same $r_1$-rank group is 0
% , and some $r_1$-rank
% differences may be determined by a successor node that has a lower
% highest level in the dual skip list that a predecessor node
.)
By \Cref{lem:sum},
$X$ is $O(n)$ with (very) high probability.
% and a similar argument applies to the sum of $r_1$-rank differences determined
% by successor nodes.
Thus, with (very) high probability, the sum of all $r_1$-rank differences
between children and parents in $T$ is $O(n)$.
Note that the root itself still requires $O(\log{\log{n}})$ bits.

Let us next consider all the $r_2$-ranks in a JIT zip-zip tree.
Recall that each time there is a rank tie when using existing $(r_1,r_2)$ 
ranks, during a given update, we augment the two $r_2$ ranks bit by bit
until they are different. 
That is, the length of each such augmentation 
is a geometric random variable with success probability $1/2$.
Further, by the way that the zip and unzip operations work,
the number of such encounters that could possibly have a rank tie
is upper bounded by the sum of the $r_1$-ranks of the keys involved, i.e., 
by the sum of
$n$ geometric random variables with success probability $1/2$.
Thus, by \Cref{lem:sum}, the number of such encounters is at most 
$N=12n$ and the number of added bits that occur
during these encounters is at most $12N$, with (very) high probability.
\qed\end{proof}

\begin{remark}The expected average number of bits of metadata per node in an RBST is also $O(1)$ if the bits are dynamically allocated. This was observed by Xiaoyang Xu (private communication, 2023).
\end{remark}
\subsection{External Zip-zip Trees}

Zip-zip trees as we have defined them use the \emph{internal} representation of a binary search tree, as do zip trees.  An alternative that is useful in some applications, e.g., Merkle trees~\cite{6233691}, is the \emph{external} representation, in which the items are stored in the external nodes of the tree, and the internal nodes contain only keys, used to guide searches.  It is straightforward to use the external representation, as we now briefly describe.  One important point is that there is one less internal node than external node.  If we want to preserve strong history independence, we must choose a unique item whose key is not in an internal node.  In our version this is the item with smallest key, but it could be the item with largest key instead.

Ignoring the question of ranks, an \emph{external} binary search tree contains a set of items in its external nodes, one item per node.  
% Each item has a key, which we assume is distinct for each item.  
We assume each item has a distinct key.
The items are in symmetric order by increasing key: If external node $x$ precedes external node $y$ in symmetric order, the key of the item in $x$ is smaller than the key of the item in $y$.  Each internal node contains the key of the item in the next node in symmetric order, which is the smallest node in symmetric order in its right subtree, reached by starting at the right child and proceeding through left children until reaching an external node.  (See Figure~\ref{fig:external-zip}.)  The item of smallest key is the unique item whose key is not stored in an internal node.  As a special case, a tree containing only one item consists of a single external node containing that item.  Instead of storing keys in internal nodes, we can store pointers to the corresponding external nodes.  Searches proceed down from the root as in an internal binary search tree, but do not stop until reaching an external node (although searches can sometimes be sped up if pointers instead of keys are stored in internal nodes).

An \emph{external zip-zip tree} is an external binary search tree in which each internal node has a rank generated as described for zip-zip trees, with the internal nodes max-heap ordered by rank and ties broken in favor of smaller key.

An insertion into a non-empty external zip-zip tree inserts two nodes into the tree, one external, containing the new item, and one internal.  To insert a new item, we generate a random rank for the new internal node.  We proceed down from the root along the search path for the key of the new item, until reaching an external node or reaching an internal node whose rank is less than that of the new internal node (including tie-breaking).  Let $x$ be the node reached, and let $y$ be its parent.  We unzip the search path from node $x$ down, splitting it into two paths, $P$, containing all nodes on the path with key less than that of the new key, and $Q$, containing all nodes on the path with key greater than that of the new key. 
% nodes on P (for example 2', 5', 5) are in increasing order I think.
% I am not sure if 'left path' or 'right path' are clear,
% but perhaps 'right spine of left subtree' is a bit verbose
% Open to ideas
Nodes on $P$ going down are in increasing order by key, so $P$ becomes a left path; those on $Q$ going down are in decreasing order by key, so $Q$ becomes a right path.  If $x$ was previously the root of the tree, the new internal node becomes the new root; otherwise, the new internal node replaces $x$ as a child of $y$.  The top node of $P$ becomes the left child of the new internal node.  The new external node becomes the left child of the bottom node of $Q$, and the top node of $Q$ becomes the right child of the new internal node.  There is one important exception: If the bottom node of $Q$ is an external node (before the new external node is added), the new external node becomes the left child of the new internal node, and the key of the bottom node on $Q$ becomes the key of the new internal node: In this case, the bottom node on $Q$ contains the item of previously smallest key, and the new item has even smaller key.  The following lemma implies that this insertion algorithm is correct:

\begin{lemma}
If an insertion results in a path $Q$ whose bottom node, say $z$, is external, then path $P$ is empty, $z$ has the smallest key before the insertion, and the key of the new item is less than that of $z$. 
\end{lemma}
\begin{proof}
Suppose $z$ is external.  If $z$ did not have the smallest key before the insertion, then in the tree before the insertion there is an internal node that is an ancestor of $z$ and contains the same key.  Let this node be $w$.  The search for the new key visits $w$ and must proceed to the left child of $w$, since since $z$ is on $Q$ and hence must have smaller key than the new key.  But $z$ is in the right subtree of $w$ and hence cannot be on the search path for the new key, a contradiction.  It follows that $z$ has the smallest key before the insertion, which further implies that $P$ is empty. 
\end{proof}

Deletion is the inverse of insertion: Search for the internal node having the key to be deleted.  Zip together the right spine of its left subtree and the left spine of its right subtree. deleting the bottom node on the latter, which is the external node whose item has the key to be deleted.  Replace the internal node having the key to be deleted by the top node on the zipped path.  If the search reaches an external node, delete this external node and its parent; replace the deleted parent by its right child.

%In an \emph{external zip-zip tree}, two nodes are added in every insertion, one internal node (the comparison node) and one external node. The internal node is given the rank $(r_1, r_2)$, where $r_1$ and $r_2$ are picked as in the canonical zip-zip tree, while the external node is not given a rank, and no rank always compares less than any defined rank. Let external nodes always be placed as the in-order successor to their internal nodes, i.e., in the left-most child of the internal node's right spine. During a deletion operation, the internal node is deleted as usual while the external node is deleted during the zipping operation. See \Cref{fig:external-zip}. Note that since the internal nodes form a canonical zip-zip tree, external zip-zip trees maintain strong history independence and support a JIT variant.

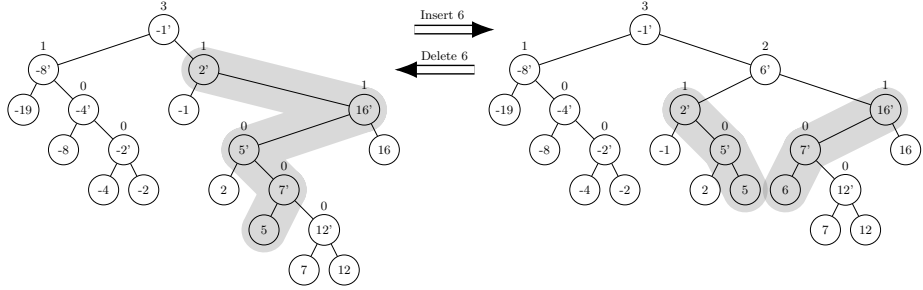
\begin{figure}[hbt]
    \centering
    \resizebox{\linewidth}{!}{\begin{tikzpicture}[n/.style = {circle, draw, minimum width = 0.75cm, minimum height=0.5cm},a/.style = {line width=1pt, double distance=5pt,
             arrows = {-Latex[length=0pt 4 .5]}}]
    \def\ziptreeone{
        {-19//2},
        {-8'/1/1},
        {-8//3},
        {-4'/0/2},
        {-4//4},
        {-2'/0/3},
        {-2//4},
        {-1'/3/0},
        {-1//2},
        {2'/1/1},
        {2//4},
        {5'/0/3},
        {5//5},
        {7'/0/4},
        {7//6},
        {12'/0/5},
        {12//6},
        {16'/1/2},
        {16//3}%
    }

    % Draw all the nodes
    \foreach \key\rank\height [count=\i] in \ziptreeone{
        \pgfmathsetmacro\x{0.5*\i};
        \node[n,label=\rank] (\key) at (\x, -\height) {\key};
    }

    % Connect the nodes as before, just make sure to swap the keys in the connection commands as well
    \draw[-] (-1') -- (-8') -- (-19)
             (-8') -- (-4') -- (-8)
             (-4) -- (-2') -- (-2)
             (-4') -- (-2')
             (12) -- (12') -- (7') -- (5') -- (16') -- (2') -- (-1')
             (7) -- (12')
             (5) -- (7')
             (2) -- (5')
             (-1) -- (2')
             (16) -- (16');

    \begin{pgfonlayer}{background}
        \fill[gray,opacity=0.3] \convexpath{2',16',5',7',5,7',5',16'}{16pt};
    \end{pgfonlayer}

    \draw [a] (10.25,0) -- node[label={[xshift=-0.3cm]Insert 6}] {} (12.25,0);
    \draw [a] (11.75,-1) -- node[label={[xshift=0.25cm]Delete 6}] {} (9.75,-1);

    % And similarly swap the keys in the definition and drawing of ziptreetwo
    \def\ziptreetwo{
        {-19//2},
        {-8'/1/1},
        {-8//3},
        {-4'/0/2},
        {-4//4},
        {-2'/0/3},
        {-2//4},
        {-1'/3/0},
        {-1//3},
        {2'/1/2},
        {2//4},
        {5'/0/3},
        {5//4},
        {6'/2/1},
        {6//4},
        {7'/0/3},
        {7//5},
        {12'/0/4},
        {12//5},
        {16'/1/2},
        {16//3}%
    }

    % Draw all the nodes for the second tree after swapping keys
    \foreach \key\rank\height [count=\i] in \ziptreetwo{
        \pgfmathsetmacro\x{0.5*\i+12};
        \node[n,label=\rank] (\key) at (\x, -\height) {\key};
    }

    % Now, draw the connections with swapped keys for the second tree
    \draw[-] (6') -- (-1') -- (-8') -- (-19)
             (-8') -- (-4') -- (-8)
             (-4) -- (-2') -- (-2)
             (-4') -- (-2')
             (5) -- (5') -- (2') -- (6') -- (16') -- (7') -- (12')
             (-1) -- (2')
             (2) -- (5')
             (7) -- (12') -- (12)
             (6) -- (7')
             (16) -- (16');

    \begin{pgfonlayer}{background}
        % Adjust the shaded areas to reflect the swapped keys if necessary
        \fill[gray,opacity=0.3] \convexpath{2',5',5,5',2'}{16pt};
        \fill[gray,opacity=0.3] \convexpath{16',7',6,7',16'}{16pt};
    \end{pgfonlayer}
\end{tikzpicture}}
    
    \caption{\label{fig:external-zip} How insertion in an external zip tree is done via unzipping and deletion is done via zipping. Comparison nodes are represented with a prime symbol. Analogous to the operation depicted in \Cref{fig:zip}. }
\end{figure}

%\begin{theorem}
%    All the keys in an external zip-zip tree are stored in the external nodes, and all comparison nodes are internal nodes.
%\end{theorem}

%\begin{proof}
%    Consider a node $n$ designated as an external node in an external zip-zip tree. Let $n'$ be the corresponding comparison node. Since the external node $n$ is placed as the left-mode node in the right spline of node $n'$, $n'$ is not an internal node. Let us assume that node $n$ is not an external node, i.e., that it contains a child node. Let's assume the child is a left child, $l$, where $l$'s key is less than the key of node $n$. Recalling that node $n$ is in the right subtree of node $n'$, node $l$'s key must be greater than that of node $n'$. Since node $n$ and node $n'$ share a key, this is a contradiction. Let's now assume that node $n$ has a right child, $r$, with a greater key. This node must be designated as an external node, since comparison nodes have ranks, and any rank compares greater than no rank. However, an external node must sit on the left spline of its comparison node's right subtree, so it cannot be the right child of a node.
%\end{proof}

\subsection{Depth Analysis}
The main theoretical result of this paper is the following.

\begin{theorem}
\label{thm:harmonic}
The expected depth, $\delta_j$, of the $j$-th smallest key 
in a zip-zip tree, $T$, storing
$n$ keys is equal to $H_j+H_{n-j+1}-1+o(1)$, where $H_n=\sum_{i=1}^n (1/i)$ is
the $n$-th harmonic number.
\end{theorem}
\begin{proof}
Let us denote the ordered list of (distinct) keys stored in $T$ as
$L=(x_1,x_2,\ldots,x_n)$,
where we use ``$x_j$'' to denote both the node in $T$ and the key
that is stored there.
% For any node $x_i$ in $T$, define the \emph{rank group} of $x_i$
% as all the nodes with the same $r_1$-rank as $x_i$ and in a subtree of 
% $T$ containing $x_i$, that is, each of these keys has a collision with $x_i$'s
% rank when just comparing the first rank coordinate, $r_1$.
Let $X$ be a random variable equal to the depth of the $j$-th 
smallest key, $x_j$, in $T$, and note that
\[
X=\sum_{i=1,\ldots,j-1,j+1,\ldots,n} X_i,
\]
where $X_i$ is an indicator random variable that is $1$ iff $x_i$ is
an ancestor of $x_j$.
Let $A$ denote the event where the $r_1$-rank of the root, $z$, of $T$
is more than $3\log n$, 
or the total size of all the
$r_1$-rank groups of $x_j$'s ancestors is more than $d\log n$, for a suitable
constant, $d$, chosen so that, by \Cref{lem:sum}, 
$\Pr(A)\le 2/n^2$.
Let $B$ denote the event, conditioned on $A$ not occurring, where
the $r_1$-rank group of an ancestor of $x_j$ contains
two keys with the same rank, i.e., their ranks are tied even
  after doing a lexicographic rank comparison.
Note that, conditioned on $A$ not occurring,
and assuming $c\ge 4$ (for the sake of a $o(1)$ additive term\footnote{Taking
  $c=3$ would only cause an $O(1)$ additive term.}),
the probability that any two keys in any 
of the $r_1$-rank groups of $x_j$'s ancestors have a tie among their
$r_2$-ranks is at most $d^2\log^2 n/\log^4 n$; hence, $\Pr(B)\le d^2/\log^2 n$.
Finally, let $C$ denote the complement event to both $A$ and $B$, that is,
the $r_1$-rank of $z$ is less than $3\log n$ and each $r_1$-rank
group for an ancestor of $x_j$ has keys with unique $(r_1,r_2)$ rank pairs.
% Then, by \Cref{thm:skip-max} and the duality between skip lists
% and zip trees, with probability at least $1-1/n$,
% there are at most $3\log n$ $r_1$-rank groups among the ancestors of $x_j$
% and the size of each one is at most $2\log n$, by Lemma~\ref{lem:rank-groups}; 
% that is, the probability
% of $B$ occurring can be bounded by $2/n$.
Thus, by the definition of conditional expectation,
\begin{eqnarray*}
\delta_j = E[X] &=& E[X|A]\cdot \Pr(A) + E[X|B]\cdot \Pr(B) + E[X|C]\cdot \Pr(C) \\
     &\le& \frac{2n}{n^2} + \frac{d^3\log n}{\log^2 n} + E[X|C] \\
     &\le& E[X|C] + o(1).
\end{eqnarray*}
So, for the sake of deriving an expectation for $X$, let us assume
that the condition $C$ holds.
Thus, for any $x_i$, where $i\not= j$, $x_i$ is an ancestor of $x_j$
iff $x_i$'s rank pair, $r=(r_1,r_2)$, is the unique maximum such rank
pair for the keys from $x_i$ to $x_j$, inclusive, in $L$
% (where we allow for either case of $x_i<x_j$ or $x_j<x_i$, and we do rank comparisons
(allowing for either case of $x_i<x_j$ or $x_j<x_i$, and doing rank comparisons
lexicographically).
Since each key in this range has equal probability of being assigned
the unique maximum rank pair among the keys in this range,
\[
\Pr(X_i=1)= \frac{1}{|i-j|+1} .
\]
Thus, 
by the linearity of expectation,
\[
E[X|C] = H_j + H_{n+1-j} -1.
\]
Therefore, $\delta_j= H_j + H_{n+1-j} -1 + o(1)$.
\qed
\end{proof}

This immediately gives us the following:

\begin{corollary}
The expected depth, $\delta_j$, of the $j$-th smallest key 
in a zip-zip tree, $T$, storing $n$ keys can be bounded as follows:
\begin{enumerate}
\item
\label{e1}
If $j=1$ or $j=n$, then 
$\delta_j < \ln n + \gamma +o(1) < 0.6932\log n+\gamma+o(1)$,
where $\gamma=0.57721566\ldots$ is the Euler-Mascheroni constant.
\item
\label{e2}
For any $1\le j\le n$, 
$\delta_j < 2\ln n -1 + o(1) < 1.3863\log n-1+o(1)$.
% \item
% \label{e3}
% For $j=\sqrt{n}$,
% $\delta_j < (3/2)\ln n + 2\gamma-1+o(1) < 1.04\log n+2\gamma-1+o(1)$.
\end{enumerate}
\end{corollary}
\begin{proof}
The bounds all follow from \Cref{thm:harmonic},
the fact that $\ln 2=0.69314718\ldots$, and 
Franel's inequality (see, e.g., Guo and Qi~\cite{GUO2011991}):
\[
H_n < \ln n + \gamma + \frac{1}{2n}.
\]
Thus, for (\myref{e1}), if $j=1$ or $j=n$, 
$\delta_j = H_n < \ln n + \gamma +o(1)$.

For (\myref{e2}), if $1\le j\le n$,
\begin{eqnarray*}
\delta_j &=& H_j+H_{n-j+1}-1 \\
         &<& \ln j + \ln (n-j+1) + 2\gamma - 1 + o(1) \\
         &\le& 2\ln n -1 + o(1),
\end{eqnarray*}
since $\ln 2 > \gamma$ and
$j(n-j+1)$ is maximized at $j=n/2$ or $j=(n+1)/2$.
%
% For (\ref{e3}), if $j=\sqrt{n}$,
% $\delta_j < H_{\sqrt{n}} + H_n -1 + o(1)
% \le (3/2)\ln n + 2\gamma -1 + o(1)$.
\qed\end{proof}

Incidentally, these bounds are actually tighter than those derived
by Seidel and Aragon for treaps~\cite{treaps}, but similar bounds
can be shown to hold for treaps.

\subsection{Height Analysis}
We similarly prove tighter bounds for the height of zip-zip trees.

\begin{theorem} The height of a zip-zip tree, $T$, holding a set, $S$, of $n$ keys is 
at most $3.82\log n$ with probability $1-o(1)$.
\end{theorem}
\begin{proof}
As in the proof of \Cref{thm:harmonic},
we note that the depth, $X$, in $T$
of the $i$-th smallest key, $x_i$, can be characterized as follows.
Let
$$
L_i = \sum_{1\le j<i} X_{j}, \mbox{~~~~and~~~~}
R_i = \sum_{i<j\le n} X_{j}, 
$$
where $X_{j}$ is a 0-1 random variable
that is 1 if and only if $x_j$ is an ancestor of $x_i$, where $x_i$
is the $i$-th smallest key in $S$ and $x_j$ is the $j$-th smallest key.
Then $X=1+L_i+R_i$.
Further, note that the random variables that are summed in $L_i$ 
(or, respectively, $R_i$)
are independent, and, focusing on $E[X|C]$, as in the proof of \Cref{thm:harmonic},
$E[L_i]=H_i-1$ and $E[R_i]=H_{n-i+1}-1$, where
$H_m=\sum_{k=1}^m 1/k$ is the $m$-th Harmonic number; hence,
$E[X|C]= H_i + H_{n-i+1}-1< 2\ln n-1$.
Thus, we can apply a Chernoff bound to characterize $X$ by
bounding $L_i$ and $R_i$ separately (w.l.o.g., we focus on $L_i$), conditioned
on $C$ holding.
For example, for the high-probability bound for the proof, it is sufficient
that, for some small constant, $\epsilon>0$,
there is a reasonably small $\delta>0$ such that
$$
P(L_i > (1+\delta)\ln n)<
2^{-((1+\epsilon)/\ln 2)(\ln 2)\log n} 
=2^{-(1+\epsilon)\log n} 
= 1/n^{1+\epsilon}, 
$$
which would establish the theorem by a union bound.
In particular, we choose $\delta=1.75$ and let $\mu=E[L_i]$. Then by a 
Chernoff bound, e.g., 
see~\cite{alon,hp,mitz,motwani,shiu}, for $\mu=\ln n$, we have the following:
\begin{eqnarray*}
\Pr(L_i>2.75\ln n )
&=& \Pr(L_i>(1+\delta )\mu ) \\
&<&\left({\frac {e^{\delta }}{(1+\delta )^{1+\delta }}}\right)^{\mu } \\
&=&\left({\frac {e^{1.75}}{2.75^{2.75}}}\right)^{\ln n } \\
&\le& 2.8^{-(\ln 2)\log n} \\
&\le& 2.04^{-\log n} \\
&=&\frac{1}{n^{\log 2.04}},
\end{eqnarray*}
which establishes the above bound for $\epsilon=\log_2 2.04-1>0$.
Combining this with a similar bound for $R_i$, and the derived from
Markov's inequality with respect to $E[X|A]$ and $E[X|B]$, given in the 
proof of \Cref{thm:harmonic} for the conditional events $A$ and $B$,
we get that the height of
of a zip-zip tree is at most $$2(2.75)(\ln 2)\log n\le 3.82\log n,$$
with probability $1-o(1)$.
\qed\end{proof}

\subsection{Making Zip-zip Trees Partially Persistent}
A data structure that can be updated in a current version while also allowing
for queries in past versions is said 
to be \emph{partially persistent},
and Driscoll, Sarnak, Sleator, and Tarjan~\cite{DRISCOLL198986}
show how to make any bounded-degree linked structure, like a binary search tree,
$T$,
into a partially persistent data structure by utilizing techniques employing
``fat nodes'' and ``node splitting.''
They show that if a sequence of $n$ updates on $T$ 
only modifies $O(n)$ data fields and pointers, then $T$ can be made partially
persistent with only an constant-factor increase in time and space for
processing the sequence of updates, and allows for queries in any past instance
of $T$.
We show below that zip-zip trees have this property, w.h.p., thereby proving
the following theorem.

\begin{theorem}
\label{thm:persistent}
One can transform an initially empty zip-zip tree, $T$, 
to be partially persistent, over the course of $n$ insert and delete
operations, so as to support, w.h.p.,
$O(\log n)$ amortized-time updates in the current version
and $O(\log n)$-time queries in the current or past versions, using
$O(n)$ space.
\end{theorem}

\begin{proof}
By the way that the zip and unzip operations work,
the total number of data or pointer changes in $T$ over the course
of $n$ insert and delete operations can be 
upper bounded by the sum of $r_1$-ranks for all the keys involved, i.e., 
by the sum of $n$ geometric random variables with success probability $1/2$.
Thus, by \Cref{lem:sum}, the number of data or pointer changes in
$T$ is at most 
$N=12n$ with (very) high probability.
Driscoll, Sarnak, Sleator, and Tarjan~\cite{DRISCOLL198986}
show how to make any bounded-degree linked structure, like a binary search tree,
$T$,
into a partially persistent data structure by utilizing techniques employing
``fat nodes'' and ``node splitting,''
so that
if a sequence of $n$ updates on $T$ 
only modifies $O(n)$ data fields and pointers, then $T$ can be made partially
persistent with only an constant-factor increase in time and space for
processing the sequence of updates, and this allows for queries 
in any past instance
of $T$ in the same asymptotic time as in the ephemeral version of $T$
plus the time to locate the appropriate prior version.
Alternatively, 
Sarnak and Tarjan~\cite{sarnak1986planar} provide a simpler set of 
techniques
that apply to binary search trees without parent parent pointers.
Combining these facts establishes the theorem.
\qed\end{proof}

For example, we can apply this theorem
with respect to a sequence of $n$ updates of a zip-zip tree
that can be performed in $O(n\log n)$ time and $O(n)$
space w.h.p., e.g., to provide a simple construction
of an $O(n)$-space planar point-location data structure
that supports $O(\log n)$-time queries. 
A similar construction was provided by 
Sarnak and Tarjan~\cite{sarnak1986planar}, based on
the more-complicated red-black tree data structure; hence,
our construction can be viewed as simplifying their construction.

\section{Experiments}

We augment our theoretical findings with experimental
results, where we repeatedly constructed search
trees with keys, $\{0,1,\ldots, n - 1\}$,
inserted in order 
(since insertion order doesn't matter).
\ifFull
Randomness was obtained by using a linear
congruential pseudo-random generator.
\fi
For both uniform zip trees and zip-zip trees with static $r_2$-ranks,
we draw integers independently for the uniform ranks from the intervals $[1, n^c]$, and $[1, \log^c n]$, respectively,
choosing $c = 3$.

% https://gcc.gnu.org/git/?p=gcc.git;a=blob;f=libstdc%2B%2B-v3/include/bits/random.h;h=42f37c1e77e127972aaac56649d4250ac368d1a6;hb=refs/heads/master

% [trim = {left, bottom, right, top}, clip]
\begin{figure}[bt!]
    \centering
\vspace*{-4pt}
    \includegraphics[width=.85\textwidth,trim={0, 0, 0, 1cm}, clip]{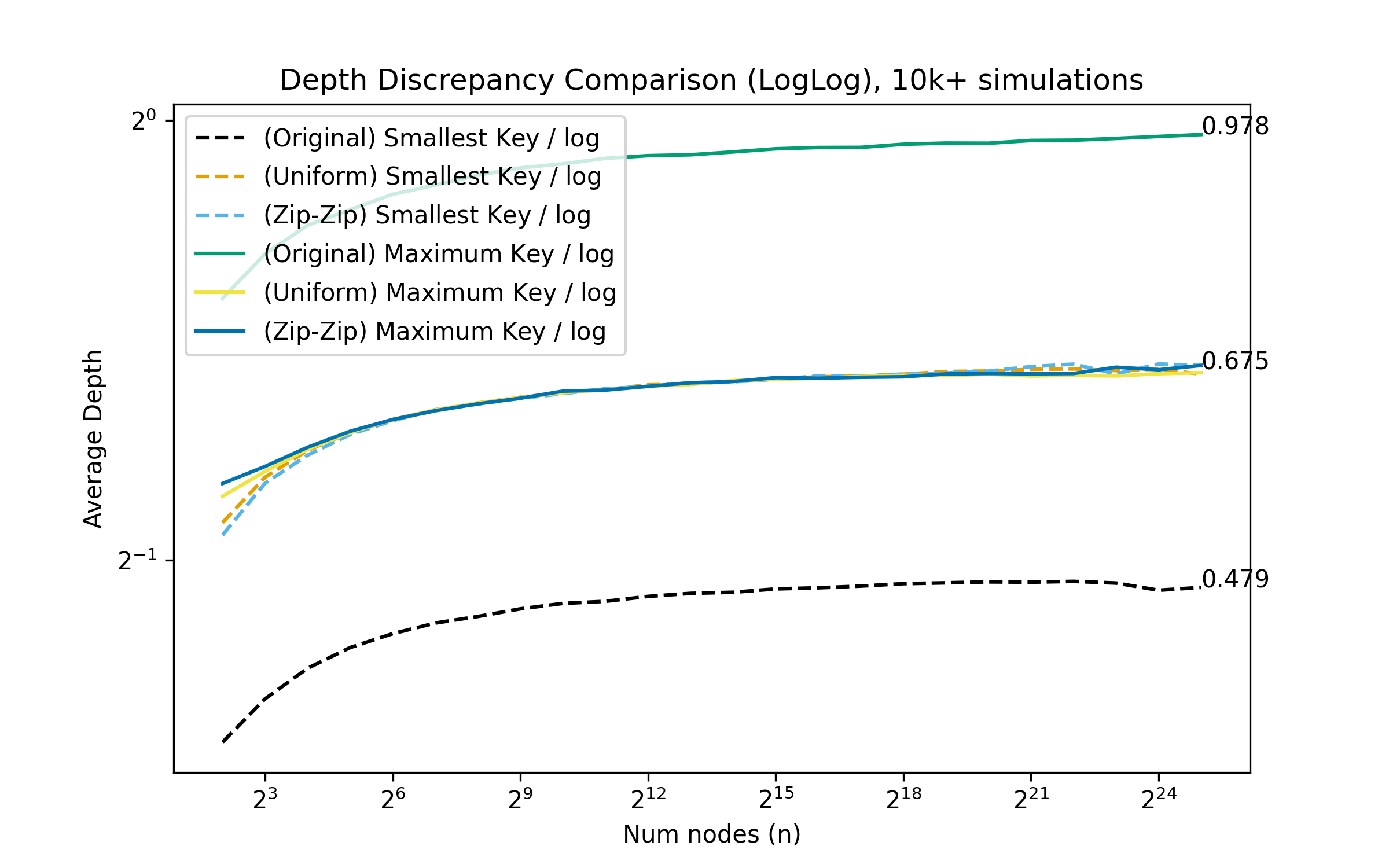}
\vspace*{-8pt}
    \caption{Experimental results for the depth discrepancy between the smallest and largest keys in the original, uniform (treap), and zip-zip variants of the zip tree. 
Each data point is scaled down by a factor of $\log n$ (base 2).}
    \label{fig:depth-discrepancy}
\end{figure}

\subsection{Depth Discrepancy}
First, we consider the respective depths of
the smallest and the largest keys in an original zip tree, compared
with the depths of these keys in a zip-zip tree.
See \Cref{fig:depth-discrepancy}.
The empirical results for the depths for 
smallest and largest keys in a zip tree clearly
match the theoretic expected values of 0.5 $\log n$ and $\log n$, respectively,
from \Cref{thm:zip-smallest-largest}. 
For comparison purposes, we
also plot the depths for smallest and largest keys 
in a uniform zip tree, which is essentially a treap,
and in a zip-zip tree (with static $r_2$-ranks).
Observe that, after the number of nodes, $n$, grows beyond
small tree sizes, there is no discernible difference between the
depths of the largest and smallest keys, and that this is very close
to the theoretical bound of $0.69\log n$.
% The uniformly distributed ranks for these experiments were
% picked from a large enough range such that rank ties did not occur
% for any reasonably sized tree, so the constructed uniform zip trees
% are topologically identical to treaps. 
\ifFull
Most notably, apart from some differences for very small trees, 
the depths for smallest and largest keys in a zip-zip
tree quickly conform to the uniform zip tree results,
while using exponentially fewer bits for each node's rank. 
\fi

\subsection{Average Key Depth and Tree Height}
Next, we empirically study the average
key depth and average height for the three aforementioned zip
tree variants. See \Cref{fig:avg-height}. 
Notably, we observe that
for all tree sizes, despite using exponentially fewer rank bits
per node, the zip-zip tree performs indistinguishably well from the
uniform zip tree, equally outperforming the original zip tree
variant.
The average key depths
and average tree heights for all variants appear to approach some
constant multiple of $\log n$. 
For example,
the average depth of a key in an original zip tree,
uniform zip tree, and zip-zip tree reached 
$1.373 \log n$,
$1.267 \log n$, 
$1.267 \log n$, 
respectively.
Interestingly, these values are roughly 8.5\% less than
the original zip tree and treap theoretical average key depths of
$1.5\log n$~\cite{zip} and $1.39\log n$~\cite{treaps}, respectively, 
suggesting that both variants approach their limits at a similar rate. 
Also, we note that our empirical 
average height bounds for uniform zip trees and zip-zip trees
get as high as $2.542\log n$. 
% It is an open problem to bound these expectations theoretically, but
% we show in an appendix that the height of a zip-zip tree is at
% most $3.82\log n$ with probability $1-o(1)$, which clearly beats
% the $4.31107\log n$ expected height for a randomly-constructed binary
% search tree~\cite{devroye1986note,devroye1987branching,flajolet1982average,reed2003height}.

% \setlength{\belowcaptionskip}{-0.5cm}

\begin{figure}[hbt]
    \centering
\vspace*{-2pt}
    \includegraphics[width=.85\textwidth,trim={0, 0, 0, 1cm}, clip]{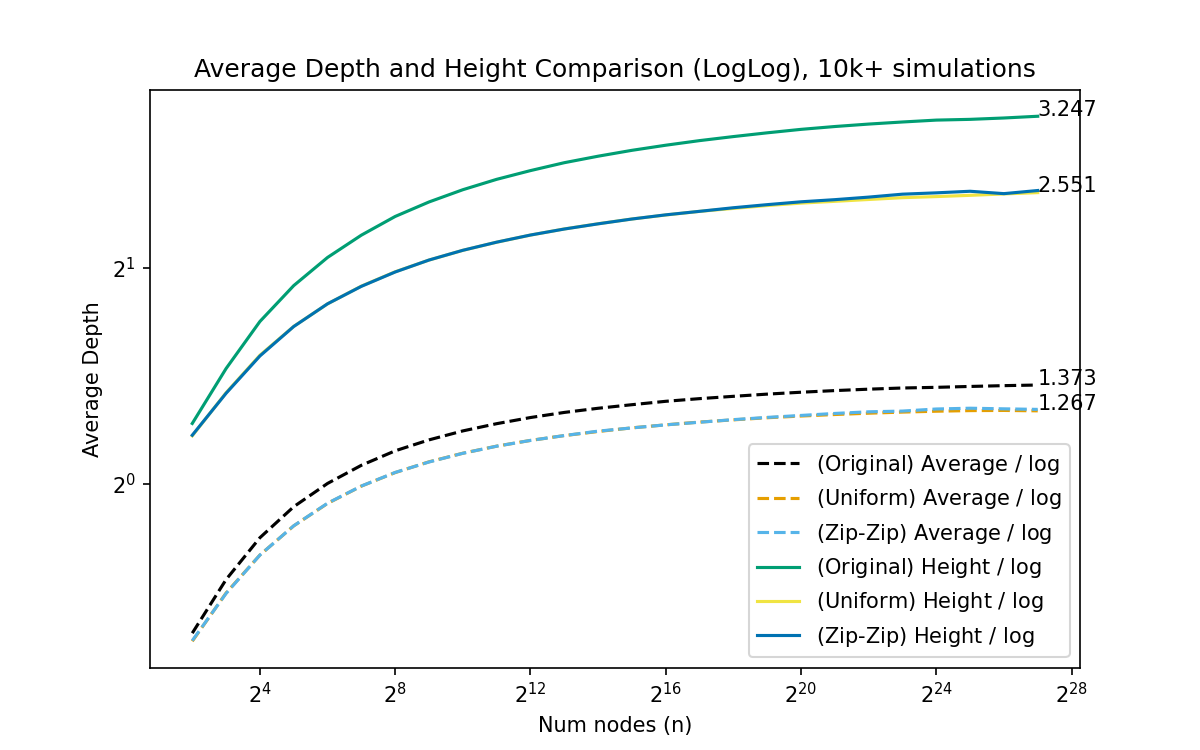}
\vspace*{-8pt}
    \caption{Experimental results for the average node depth and tree height, comparing the original, uniform (treap-like), and zip-zip variants of the zip tree. Each data point is scaled down by a factor of $\log n$ (base 2).}
    \label{fig:avg-height}
\end{figure}

% \subsection{Heights}

% \begin{figure}[hbt]
%     \centering
%     \includegraphics[width=\textwidth]{figs/comparison-height.png}
%     \caption{Experimental results for the average tree height, comparing the original, uniform (treap), and zip-zip variants of the zip tree. Each data point is scaled down by a factor of $\log n$ (base 2).}
%     \label{fig:heights}
% \end{figure}

\subsection{Rank Comparisons}
Next, we experimentally determine the frequency of complete rank
ties (collisions) for the uniform and zip-zip variants. 
% These results vary slightly depending on whether the recursive or sequential
% implementations are used. For these results, we used the sequential
% implementations.
% 
% When a node added to the original zip tree experiences a collision, it has a 50/50 chance of having a greater
See \Cref{fig:comparisons} (left).
The experiments show how the frequencies of rank collisions decrease polynomially in $n$ for the uniform zip tree and in $\log n$ for the second rank of the zip-zip variant. This reflects how these rank values were drawn uniformly from a range of $n^c$ and $\log^c n$, respectively. Specifically, we observe the decrease to be polynomial to $n^{-2.97}$ and $\log^{-2.99}{n}$, matching our chosen value of $c$ being 3.

% % \todo{Does p = 1/4 make sense for original zip tree?} 
% The experiments show how the frequency of rank collisions for the uniform zip tree variant decrease polynomially with tree size, reflecting how the rank values were drawn uniformly from a polynomially increasing distribution. Specifically, we observe the frequency decrease being proportional to $n^{-2.99}$, while we expect it to be proportional to $n^{-c}$, or $n^{-3}$ in our case. For the zip-zip trees there are two different ranks, the geometrically distributed ranks and the uniformly distributed ranks. We expect two independently chosen geometrically distributed ranks to collide $1/3$ of the time, and we observe a collision rate of 0.307 (perhaps rank comparisons occur between nodes whose ranks are not fully independent) \todo{investigate}.\footnote{$P(\text{geometric collision}) = \sum_{k = 1}^\infty{\left (2^{-k} \right )^2} = 1/3$} Similarly, we expect both ranks to collide for two independently chosen nodes to be $1/3 \log^{-c}{n}$. When $n = 2^{24}$, we would consequently expect the frequency of rank collisions (assuming independence) to be $2.41 \times 10^{-5}$; we observe $2.53 \times 10^{-5}$.

\begin{figure}[b!]
    \centering
% \vspace*{-2pt}
\hspace*{-16pt}
\begin{minipage}{1.1\textwidth}
    \includegraphics[width=.525\textwidth,trim={0, 0, 0, 1cm}, clip]{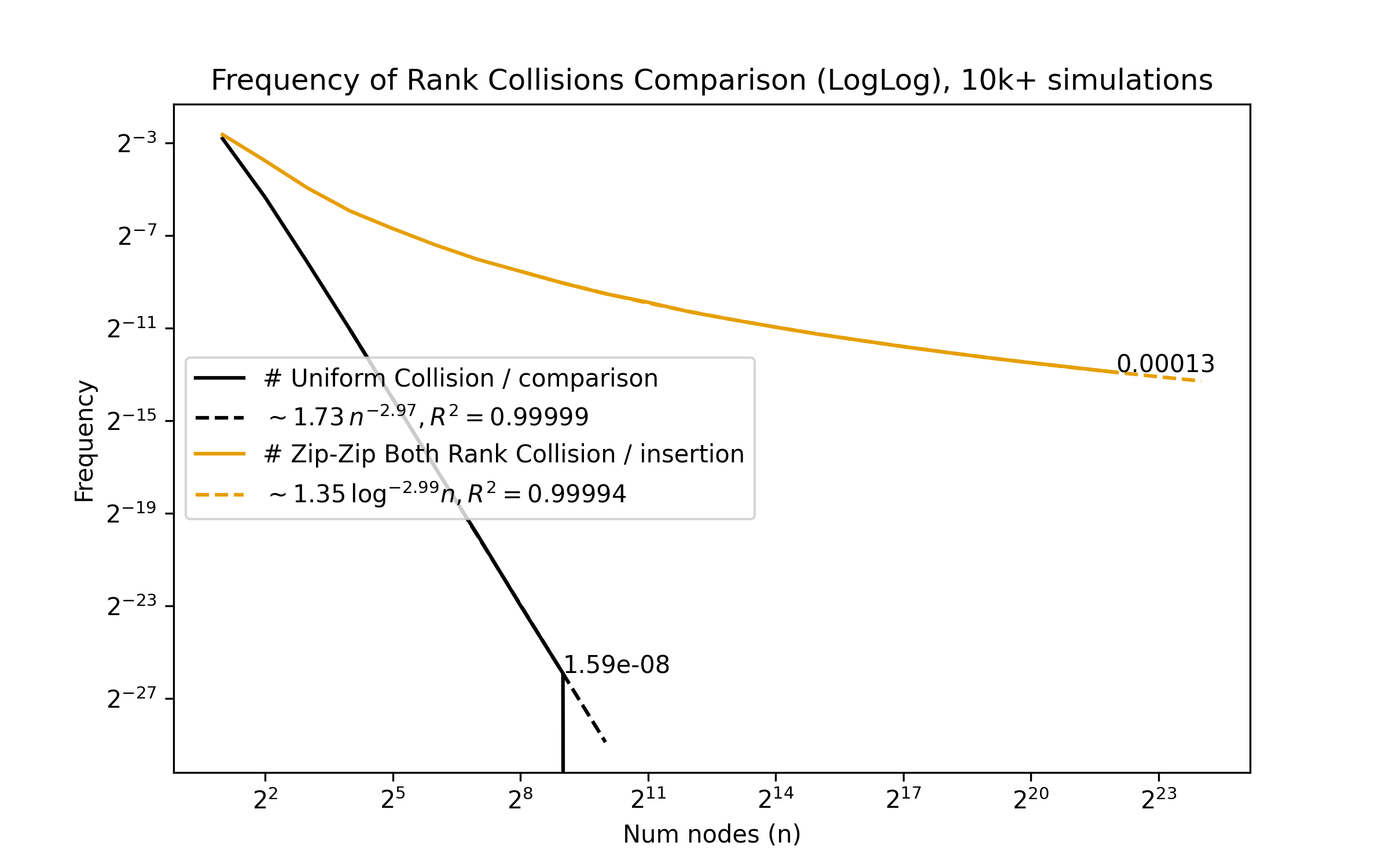}
\hspace*{-24pt}
    \includegraphics[width=.525\textwidth,trim={0, 0, 0, 1cm}, clip]{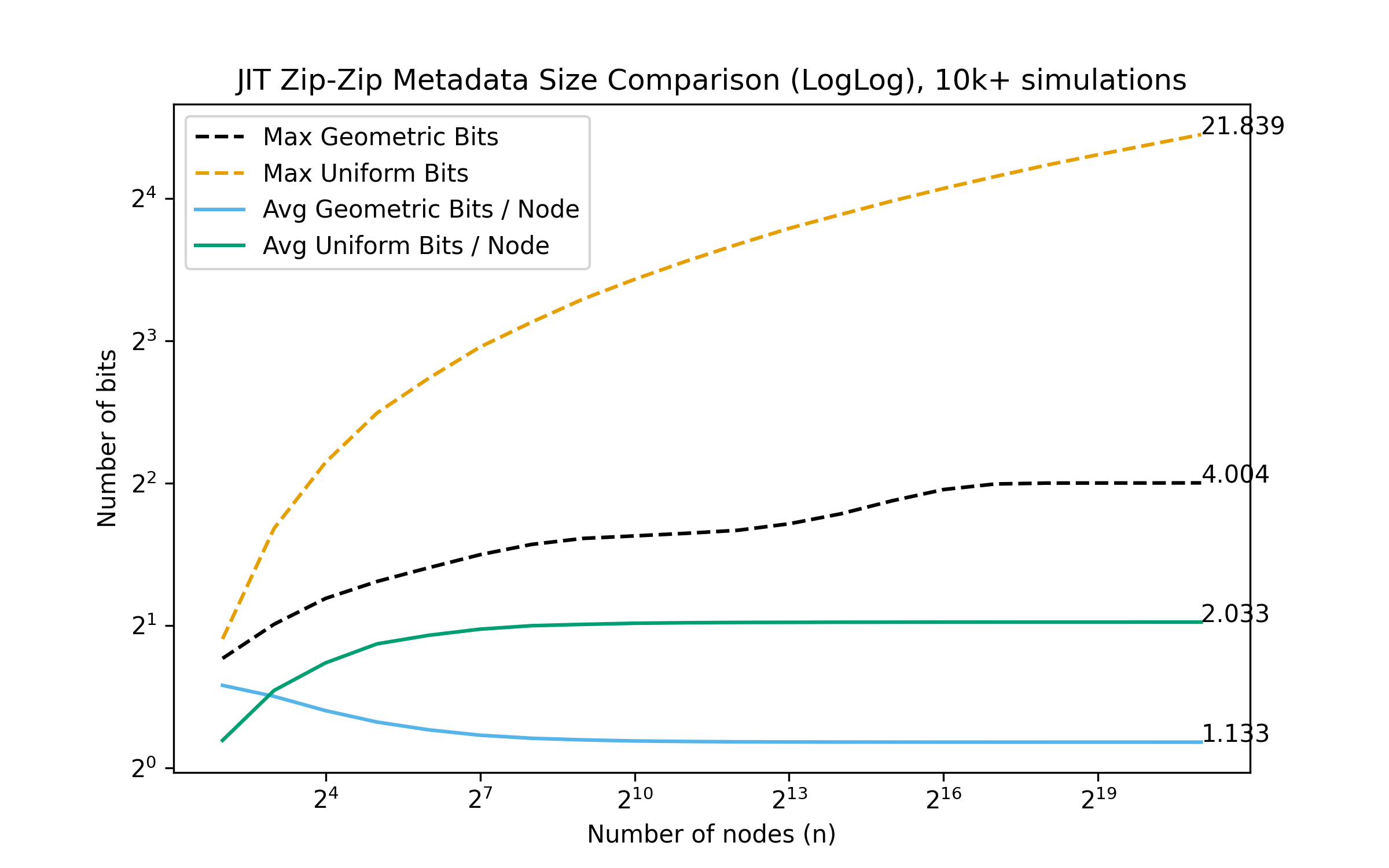}
\end{minipage}
% \vspace*{-8pt}
    \caption{(Left) The frequency of encountered rank ties per rank comparison for the uniform variant and per element insertion for the zip-zip variant.
    (Right) The metadata size for the just-in-time implementation of the zip-zip tree.
}
    \label{fig:jit-metadata}
    \label{fig:comparisons}
\end{figure}

% \begin{figure}
%     \centering
%     \begin{minipage}{.5\textwidth}
%         \centering
%         \includegraphics[width=.9\textwidth,trim={0, 0, 0, 1cm}, clip]{figs/sequential-collision-comparisons.png}
%         \captionof{figure}{The frequency of encountered rank ties (collisions) per rank comparison for the uniform variant and per element insertion for the zip-zip variant.}
%         \label{fig:comparisons}
%     \end{minipage}%
%     \hfill
%     \begin{minipage}{.5\textwidth}
%         \centering
%         \includegraphics[width=.9\textwidth,trim={0, 0, 0, 1cm}, clip]{figs/jit-size-comparison.png}
%         \captionof{figure}{The metadata size for the just-in-time implementation of the zip-zip tree.}
%         \label{fig:jit-metadata}
%     \end{minipage}
% \end{figure}

\subsection{Just-in-Time Zip-zip Trees}

In our final zip-zip tree experiment, we show how the just-in-time variant uses an
expected constant number of bits per node. 
See \Cref{fig:jit-metadata} (right).
We observe a 
results of only $1.133$ bits per node for storing the geometric
($r_1$) rank differences, and only $2.033$ bits per node for storing
the uniform ($r_2$) ranks, leading to a remarkable total of $3.166$
expected bits per node of rank metadata to achieve ideal treap
properties. Note that these results were obtained when nodes were inserted in increasing order of keys, and may not hold in general. For a uniformly at random insertion order, results were largely similar.

\subsection{Varying Geometric Mean}\label{subsec:variable-p}
% Ofek: Maybe rename to Varying Success Probability
%                    or Biased Coins

In the original zip tree paper, the authors suggest that zip trees could be more balanced by increasing the mean of the geometric distribution by which a nodes' rank is chosen. The authors left this question open to experimental study, which we will now address.

We ran our experiments using zip trees with $2^{16}$ or around 65 thousand keys, varying the success probability of the geometric distribution from 0.00001 to 0.999, which in turn varies the mean from 10,000 to $1.\overline{001}$. Recall that the original zip tree reaches an average depth of 1.30 $\log{n}$ and height of 2.96 $\log{n}$ while using roughly $1 \log{\log{n}}$ bits of space and that the zip-zip tree reaches an average depth of 1.21 $\log{n}$ and height of 2.37 $\log{n}$ while using roughly $4 \log{\log{n}}$ bits of space. \Cref{fig:variable-p} confirms the results for the original zip trees, perfectly matching depth, height, and memory results when $p = 1/2$. Interestingly, when $p = 0.0002$ the depth, height, and memory results of this modified zip tree perfectly match results from the new zip-zip tree.

\begin{figure}[tb!]
    \centering
    \vspace*{-2pt}
    \hspace*{-16pt}
    \begin{minipage}{1.1\textwidth}
        \includegraphics[width=.525\textwidth,trim={0, 0, 0, 1cm}, clip]{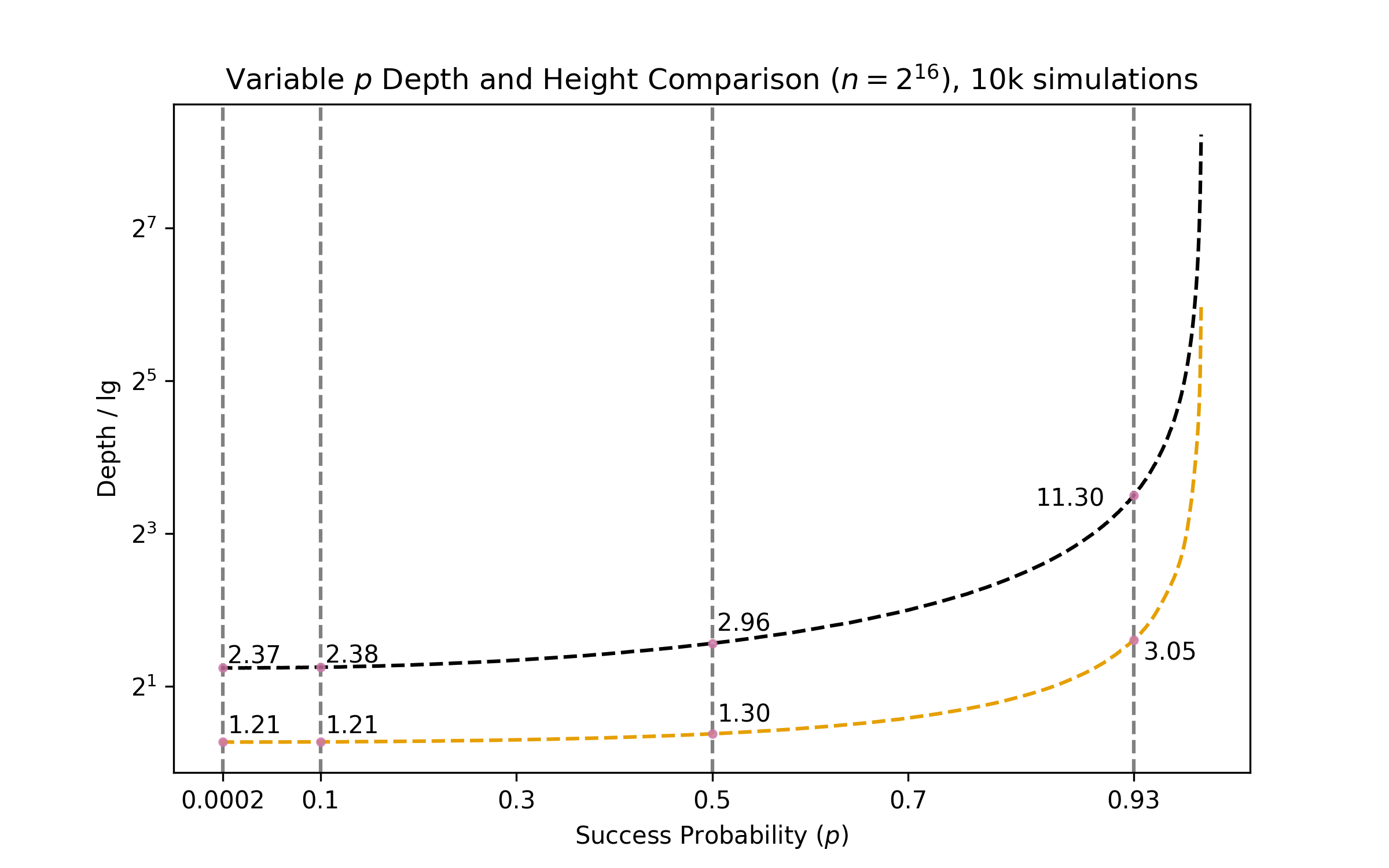}
    \hspace*{-24pt}
        \includegraphics[width=.525\textwidth,trim={0, 0, 0, 1cm}, clip]{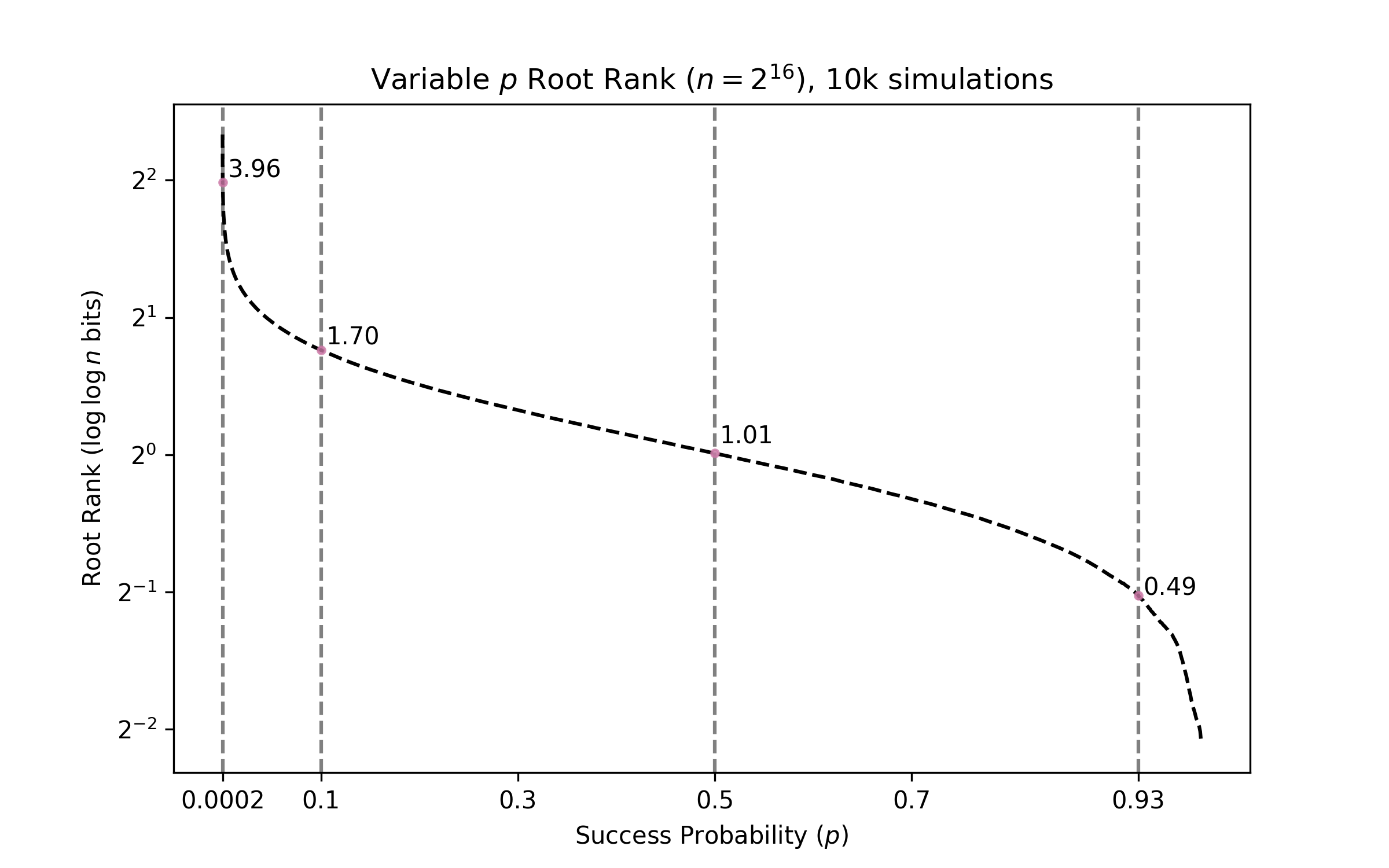}
    \end{minipage}
    % \vspace*{-8pt}
        \caption{Experimental results for the original zip tree when varying the geometric success probability ($p$) for the rank distribution. These show the trade-off between the number of bits required (Right) versus the performance gained (Left). Like before, the depths and heights are scaled down by a factor of $\log{n}$, while the root ranks this time are scaled by a factor of $\log{\log{n}}$ (all base 2).
    }
    \label{fig:variable-p}
\end{figure}

\section{Biased Zip-zip Trees}
In this section, we describe how to make zip-zip trees biased  
for weighted keys.
In this case, we assume each key, $k$, has an associated weight, $w_k$, such
as an access frequency.
Without loss of generality, we assume that weights don't change, since we can
simulate a weight change by deleting and reinserting a key with its new
weight.

Our method for
modifying zip-zip trees to accommodate weighted keys is simple---when we
insert a key, $k$, with weight, $w_k$, we now 
assign $k$ a rank pair, $r=(r_1,r_2)$,
such that $r_1$ is $\lfloor\log w_k\rfloor + X_k$, where $X_k$ is drawn
independently
from a geometric distribution with success probability $1/2$, and $r_2$
is an integer independently chosen uniformly in the range from $1$ to
$\lceil\log^c n\rceil$, where $c\ge 3$.
Thus, the only modification to our zip-zip tree construction to define
a biased zip-zip tree is that 
the $r_1$ component is now a sum of a logarithmic rank and a value drawn from
a geometric distribution. 
As with our zip-zip tree definition for unweighted keys,
all the update and search operations for biased zip-zip trees are the same
as for the original zip trees, except for this modification to the rank, $r$,
for each key (and performing rank comparisons lexicographically).
Therefore, assuming polynomial weights, we still
can represent each such rank, $r$, using $O(\log\log n)$ bits w.h.p.

We also have the following theorem, which implies the 
expected search performance bounds for weighted keys.

\begin{theorem}
The expected depth of a key, $k$, with weight, $w_k$, in a biased zip-zip
tree storing a set, $\mathcal{K}$, of $n$ keys is $O(\log (W/w_k))$,
where $W=\sum_{k\in\mathcal{K}} w_k$.
\end{theorem}
\begin{proof}
By construction, a biased zip-zip tree, $T$, is dual to a 
biased skip list, $L$, 
defined on $\mathcal{K}$
with the same $r_1$ ranks as for the keys in $\mathcal{K}$ as assigned
during their insertions into $T$.
Bagchi, Buchsbaum, and Goodrich~\cite{bagchi2005biased} show that the
expected depth of a key, $k$, in $L$ is $O(\log (W/w_k))$.
Therefore, by \Cref{thm:skip-max},
and the linearity of expectation,
the expected depth of $k$ in $T$ is $O(\log (W/w_k))$, where, as mentioned
above, $W$ is the sum of the weights of the keys in $T$ and $w_k$
is the weight of the key, $k$.
\qed
\end{proof}

Thus, a biased zip-zip tree has similar expected search and update performance
as a biased skip list, but with reduced space, since a biased zip-zip tree
has exactly $n$ nodes, whereas, assuming a standard skip-list representation
where we use a linked-list node for each instance of a key, $k$, on a level
in the skip list (from level-0 to the highest level where $k$ appears)
a biased skip list has an expected number
of nodes equal to $2n+2\sum_{k\in\mathcal{K}} \log w_k$.
For example, if there are $n^\epsilon$ keys with weight $n^\epsilon$, then
such a biased skip list would require $\Omega(n\log n)$ nodes, whereas
a dual biased zip-zip tree would have just $n$ nodes.

Further, due to their simplicity and weight biasing,
we can utilize biased zip-zip trees as 
the biased auxiliary data structures in the link-cut dynamic tree
data structure of Sleator and Tarjan~\cite{link-cut}, thereby providing
a simple implementation of link-cut trees.

\section{Future Work}

In our paper, there is a clear trade-off between memory and history independence.
In order to achieve an expected constant amount of metadata bits per node, history independence must be sacrificed.
It remains interesting to see whether a version of the zip tree that is able to optimize for both while still maintaining good average node depth and height exists.
In \Cref{subsec:variable-p} we ran experiments on a version of the zip tree where the geometric mean was increased and saw that it was able to reproduce the results of the zip-zip tree. While such a variant would not be able to run using only an expected constant number of bits per node in the same way as the JIT variant, it nevertheless remains interesting whether something can be proved about the average node depth and height of such a tree.
Particularly whether it is possible to similarly achieve good asymptotic bounds if the geometric mean is some function of the final size of the tree.
As stated in the original paper, the zip tree and its variants present themselves well to concurrent implementations, and there remains no known non-blocking implementation.
% maybe mention something about a cache-efficient version

\section{Declarations}

\subsection{Conflict of Interest}
The authors declare that there were no conflicts of interest during the writing and publication of these results.

\clearpage
    \bibliographystyle{splncs04}
    \bibliography{refs}

\begin{thebibliography}{10}
\providecommand{\url}[1]{\texttt{#1}}
\providecommand{\urlprefix}{URL }
\providecommand{\doi}[1]{https://doi.org/#1}

\bibitem{acar}
Acar, U.A.: Self-Adjusting Computation. Ph.D. thesis, Carnegie Mellon Univ. (2005)

\bibitem{afek_cb_2014}
Afek, Y., Kaplan, H., Korenfeld, B., Morrison, A., Tarjan, R.E.: The {CB} tree: a practical concurrent self-adjusting search tree  \textbf{27}(6),  393--417. \doi{10.1007/s00446-014-0229-0}

\bibitem{alon}
Alon, N., Spencer, J.H.: The Probabilistic Method. John Wiley \& Sons, 4th edn. (2016)

\bibitem{bagchi2005biased}
Bagchi, A., Buchsbaum, A.L., Goodrich, M.T.: Biased skip lists. Algorithmica  \textbf{42},  31--48 (2005)

\bibitem{tiny-pointers}
Bender, M.A., Conway, A., Farach-Colton, M., Kuszmaul, W., Tagliavini, G.: Tiny pointers. In: ACM-SIAM Symposium on Discrete Algorithms (SODA). pp. 477--508 (2023). \doi{10.1137/1.9781611977554.ch21}

\bibitem{bent1985biased}
Bent, S.W., Sleator, D.D., Tarjan, R.E.: Biased search trees. SIAM Journal on Computing  \textbf{14}(3),  545--568 (1985)

\bibitem{duality}
Dean, B.C., Jones, Z.H.: Exploring the duality between skip lists and binary search trees. In: Proc. of the 45th Annual Southeast Regional Conference (ACM-SE). pp. 395--399 (2007). \doi{10.1145/1233341.1233413}

\bibitem{devroye1986note}
Devroye, L.: A note on the height of binary search trees. J. ACM  \textbf{33}(3),  489--498 (1986)

\bibitem{devroye1987branching}
Devroye, L.: Branching processes in the analysis of the heights of trees. Acta Informatica  \textbf{24}(3),  277--298 (1987)

\bibitem{DRISCOLL198986}
Driscoll, J.R., Sarnak, N., Sleator, D.D., Tarjan, R.E.: Making data structures persistent. Journal of Computer and System Sciences  \textbf{38}(1),  86--124 (1989). \doi{10.1016/0022-0000(89)90034-2}

\bibitem{eberl2018verified}
Eberl, M., Haslbeck, M.W., Nipkow, T.: Verified analysis of random binary tree structures. In: 9th Int. Conf. on Interactive Theorem Proving (ITP). pp. 196--214. Springer (2018)

\bibitem{erickson_treaps_2017}
Erickson, J.: Lecture notes on treaps. Online (2017), available: \url{https://jeffe.cs.illinois.edu/teaching/algorithms/notes/03-treaps.pdf}

\bibitem{flajolet1982average}
Flajolet, P., Odlyzko, A.: The average height of binary trees and other simple trees. Journal of Computer and System Sciences  \textbf{25}(2),  171--213 (1982)

\bibitem{goodrich2015algorithm}
Goodrich, M.T., Tamassia, R.: Algorithm Design and Applications. Wiley (2015)

\bibitem{GUO2011991}
Guo, B.N., Qi, F.: Sharp bounds for harmonic numbers. Applied Mathematics and Computation  \textbf{218}(3),  991--995 (2011). \doi{10.1016/j.amc.2011.01.089}

\bibitem{hp}
Hagerup, T., Rüb, C.: A guided tour of {Chernoff} bounds. Information Processing Letters  \textbf{33}(6),  305--308 (1990)

\bibitem{hartline}
Hartline, J.D., Hong, E.S., Mohr, A.E., Pentney, W.R., Rocke, E.C.: Characterizing history independent data structures. Algorithmica  \textbf{42},  57--74 (2005)

\bibitem{rbst}
Mart\'{\i}nez, C., Roura, S.: Randomized binary search trees. J. ACM  \textbf{45}(2),  288--323 (1998). \doi{10.1145/274787.274812}

\bibitem{6233691}
Merkle, R.C.: Protocols for public key cryptosystems. In: 1980 IEEE Symposium on Security and Privacy. pp. 122--122 (1980). \doi{10.1109/SP.1980.10006}

\bibitem{mitz}
Mitzenmacher, M., Upfal, E.: Probability and Computing: Randomization and Probabilistic Techniques in Algorithms and Data Analysis. Cambridge University Press, 2nd edn. (2017)

\bibitem{motwani}
Motwani, R., Raghavan, P.: Randomized Algorithms. Cambridge University Press (1995)

\bibitem{papadakis1992average}
Papadakis, T., Ian~Munro, J., Poblete, P.V.: Average search and update costs in skip lists. BIT Numerical Mathematics  \textbf{32}(2),  316--332 (1992)

\bibitem{skip-lists}
Pugh, W.: Skip lists: A probabilistic alternative to balanced trees. Commun. ACM  \textbf{33}(6),  668--676 (jun 1990). \doi{10.1145/78973.78977}

\bibitem{reed2003height}
Reed, B.: The height of a random binary search tree. J. ACM  \textbf{50}(3),  306--332 (2003)

\bibitem{sarnak1986planar}
Sarnak, N., Tarjan, R.E.: Planar point location using persistent search trees. Communications of the ACM  \textbf{29}(7),  669--679 (1986)

\bibitem{treaps}
Seidel, R., Aragon, C.R.: Randomized search trees. Algorithmica  \textbf{16}(4-5),  464--497 (1996)

\bibitem{shiu}
Shiu, D.: Efficient computation of tight approximations to {Chernoff} bounds. Computational Statistics pp. 1--15 (2022)

\bibitem{link-cut}
Sleator, D.D., Tarjan, R.E.: A data structure for dynamic trees. In: 13th ACM Symposium on Theory of Computing (STOC). pp. 114--122 (1981)

\bibitem{zip}
Tarjan, R.E., Levy, C., Timmel, S.: Zip trees. ACM Trans. Algorithms  \textbf{17}(4),  34:1--34:12 (2021). \doi{10.1145/3476830}

\end{thebibliography}

\clearpage
\begin{appendix}

\section{Pseudo-code for Insertion and Deletion in Zip Trees and Zip-zip Trees}\label{sec:pseudo}
For completeness, we give the pseudo-code for the insert and delete operations,
from Tarjan, Levy, and Timmel~\cite{zip}, in \Cref{fig:insert,fig:delete}.

\begin{figure}[hbtp]
    \centering

    \begin{adjustwidth}{-0.5cm}{-0.5cm}
    \begin{algorithmic}[]
        \Function{Insert}{$x$}
            \State $rank \gets x.rank \gets$ \Call{RandomRank}{}
            \State $key \gets x.key$
            \State $cur \gets root$
            \While{$cur \neq \text{null}$ and $(rank < cur.rank$ or $(rank = cur.rank$ and $key > cur.key))$}
                \State $prev \gets cur$
                \State $cur \gets \IF key < cur.key \THEN cur.left \ELSE cur.right$
            \EndWhile
            \vspace{0.5\baselineskip}
            \State $\IF cur = root \THEN root \gets x$
            \State $\ELIF key < prev.key \THEN prev.left \gets x$
            \State $\textbf{else}~ prev.right \gets x$
            \Statex
            \State $\IF cur = \text{null} \THEN \{~x.left \gets x.right \gets \text{null}; ~\Return~ \}$
            \State $\IF key < cur.key \THEN x.right \gets cur \ELSE x.left \gets cur$
            \State $prev \gets x$
            \Statex
            \While{$cur \ne \text{null}$}
                \State $fix \gets prev$
                \Statex
                \If{$cur.key < key$}
                    \Repeat{$~\{~prev \gets cur;~cur \gets cur.right~\}$}
                    \Until{$cur = \text{null}$ or $cur.key > key$}
                \Else
                    \Repeat{$~\{~prev \gets cur;~cur \gets cur.left~\}$}
                    \Until{$cur = \text{null}$ or $cur.key < key$}
                \EndIf
                \vspace{0.5\baselineskip}
                \If{$fix.key > key$ or $(fix = x$ and $prev.key > key)$}
                    \State $fix.left \gets cur$
                \Else
                    \State $fix.right \gets cur$
                \EndIf
            \EndWhile
        \EndFunction
    \end{algorithmic}
    \end{adjustwidth}

    \caption{\label{fig:insert} Insertion in a zip tree (or zip-zip tree), 
from~\cite{zip}.}
\end{figure}

\begin{figure}[hbtp]
    \centering
    \begin{adjustwidth}{-0.5cm}{-0.5cm}
        \begin{algorithmic}[]
            \Function{Delete}{$x$}
                \State $key \gets x.key$
                \State $cur \gets root$
                \While{$key \ne cur.key$}
                    \State $prev \gets cur$
                    \State $cur \gets \IF key < cur.key \THEN cur.left \ELSE cur.right$
                \EndWhile
                \vspace{0.5\baselineskip}
                \State $left \gets cur.left;~right \gets cur.right$
                \Statex
                \State $\IF left = \text{null} \THEN cur \gets right$
                \State $\ELIF right = \text{null} \THEN cur \gets left$
                \State $\ELIF left.rank \ge right.rank \THEN cur \gets left$
                \State $\textbf{else}~ cur \gets right$
                \Statex
                \State $\IF root = x \THEN root \gets cur$
                \State $\ELIF key < prev.key \THEN prev.left \gets cur$
                \State $\textbf{else}~ prev.right \gets cur$
                \Statex
                \While{$left \ne \text{null}$ and $right \ne \text{null}$}
                    \If{$left.rank \ge right.rank$}
                        \Repeat{$~\{~prev \gets left;~left \gets left.right~\}$}
                        \Until{$left = \text{null}$ or $left.rank < right.rank$}
                        \State $prev.right \gets right$
                    \Else
                        \Repeat{$~\{~prev \gets right;~right \gets right.left~\}$}
                        \Until{$right = \text{null}$ or $left.rank \ge right.rank$}
                        \State $prev.left \gets left$
                    \EndIf
                \EndWhile
            \EndFunction
        \end{algorithmic}
    \end{adjustwidth}

    \caption{\label{fig:delete} Deletion in a zip tree (or zip-zip tree), 
    from~\cite{zip}.}
\end{figure}

\end{appendix}

\end{document}